\documentclass{article}
\pdfoutput=1

\usepackage{microtype}
\usepackage{graphicx}
\usepackage{subcaption}
\usepackage{booktabs} 

\usepackage{hyperref}

\usepackage{xcolor}

\usepackage[accepted]{icml2024_arxiv}

\usepackage{amssymb}
\usepackage{mathtools}
\usepackage{amsthm}

\theoremstyle{plain}
\newtheorem{theorem}{Theorem}[section]
\newtheorem{proposition}[theorem]{Proposition}

\theoremstyle{definition}
\newtheorem{definition}[theorem]{Definition}

\theoremstyle{remark}

\usepackage{siunitx}
\usepackage[utf8]{inputenc} %
\usepackage[T1]{fontenc}    %
\usepackage{url}            %
\usepackage{amsfonts}       %
\usepackage{nicefrac}       %
\usepackage{xcolor}         %
\usepackage{amsmath}
\usepackage{physics}
\usepackage{nameref}
\usepackage{algorithm}
\usepackage{algorithmic}
\usepackage{etoolbox}
\usepackage{caption}
\usepackage{subcaption}
\usepackage{xfrac}
\usepackage{multirow}
\usepackage{array}
\usepackage{calc}
\usepackage{svg}
\usepackage{wrapfig}
\AtBeginDocument{\RenewCommandCopy\qty\SI}
\DeclareSIUnit{\nothing}{\relax}
\usepackage{colortbl}
\usepackage[nolist]{acronym}
\newcolumntype{V}{>{\centering\arraybackslash}m{0.4\linewidth} }
\setcitestyle{citesep={;}}

\graphicspath{{./resources/}{./resources/out_imgs/}}

\makeatletter
\providecommand*{\input@path}{}
\g@addto@macro\input@path{{figures/}{./resources/out_imgs/}}%
\makeatother

\acrodef{RNN}[RNN]{recurrent neural network}
\acrodef{CNN}[CNN]{convolutional neural network}
\acrodef{MEA}[MEA]{multi-electrode array}
\acrodef{RGC}[RGC]{retinal ganglion cell}
\acrodef{SDF}[SDF]{signed distance function}
\acrodef{IQM}[IQM]{interquartile mean}
\acrodef{GLM}[GLM]{generalized linear model}

\begin{document}

\twocolumn[
\icmltitle{Spike Distance Function as a Learning Objective for Spike Prediction}

\begin{icmlauthorlist}
\icmlauthor{Kevin Doran}{lifesci,informatics}
\icmlauthor{Marvin Seifert}{lifesci}
\icmlauthor{Carola A. M. Yovanovich}{lifesci}
\icmlauthor{Tom Baden}{lifesci,ior}
\end{icmlauthorlist}

\icmlaffiliation{lifesci}{School of Life Sciences, University of Sussex, UK}
\icmlaffiliation{informatics}{School of Engineering and Informatics, University of Sussex, UK}
\icmlaffiliation{ior}{Institute of Ophthalmic Research, University of Tübingen, Germany}

\icmlcorrespondingauthor{Kevin Doran}{k.a.doran@sussex.ac.uk}

\icmlkeywords{machine learning, ICML, neuroscience, spike prediction}

\vskip 0.3in
]

\begin{NoHyper}
\printAffiliationsAndNotice{}
\end{NoHyper}

\begin{abstract}
Approaches to predicting neuronal spike responses commonly use a Poisson learning objective. This objective quantizes responses into spike counts within a fixed summation interval, typically on the order of \numrange{10}{100} milliseconds in duration; however, neuronal responses are often time accurate down to a few milliseconds, and Poisson models struggle to precisely model them at these timescales. We propose the concept of a spike distance function that maps points in time to the temporal distance to the nearest spike. We show that neural networks can be trained to approximate spike distance functions, and we present an efficient algorithm for inferring spike trains from the outputs of these models. Using recordings of chicken and frog retinal ganglion cells responding to visual stimuli, we compare the performance of our approach to that of Poisson models trained with various summation intervals. We show that our approach outperforms the use of Poisson models at spike train inference.
\end{abstract}

\section{Introduction} 
This paper proposes a new learning objective for the problem of spike prediction. Spike prediction is the task of estimating the timing of future action potentials (spikes) of a neuron; for example, given \qty{1000}{\ms} of stimulus and spike activity, predict the next \qty{80}{\ms} of spike activity. The task is illustrated in Figure \ref{fig:task}. 

The design of retinal prosthetics is an example where predicting spike sequences is important: patients suffering from diseases such as retinitis pigmentosa can have their vision partially restored through electrical or optogenetic stimulation that attempts to reproduce the spike sequences of \acp{RGC} in a healthy retina \citep{bordaAdvancesVisualProstheses2022, sahelPartialRecoveryVisual2021}.

\begin{figure}
    \centering
    \captionsetup{type=figure}
    \def\svgwidth{0.8\textwidth}
    \includegraphics[width=\columnwidth]{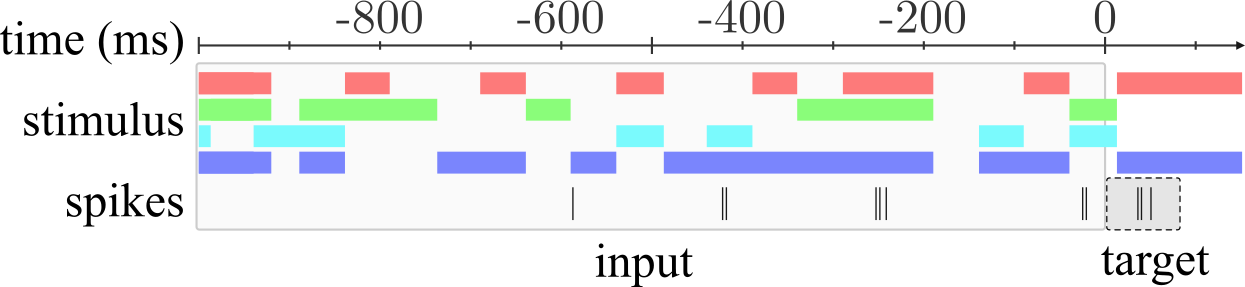}
    \captionof{figure}{\textbf{Spike prediction task:} predict future spikes given stimulus and spike history. For this work, stimulus history is a 1-second snippet of a widefield stimulus at four wavebands (\qty{420}{nm}, \qty{480}{nm}, \qty{505}{nm} and \qty{630}{nm}). The prediction duration is not fixed—several durations are explored. Information about the stimulus after $t=0$ is not available to the model, matching the task faced by retinal prosthetic devices.} 
  \label{fig:task} 
\end{figure}

A widely used approach to spike prediction is to train a model to predict the number of spikes that will occur within a fixed time interval. A probabilistic motivation is often given by describing the neuron's output as a Poisson process. In this setting, the model's output is interpreted as a firing rate and corresponds to the single parameter of the Poisson process at a certain point in time. Training the model amounts to maximizing the likelihood of this time-varying parameter with respect to the data. 

For the Poisson approach, the interval length over which spikes are summed is a hyperparameter that impacts the training and inference of a model. As will be described in Section \ref{sec:poisson}, the desire to reach millisecond resolution encourages shorter summation intervals; however, as the summation interval is reduced, training becomes more difficult, and both the effects of binning and the shape of the Poisson distribution can limit performance. 

\begin{figure*}[t]
    \centering
    \includegraphics[width=\textwidth]{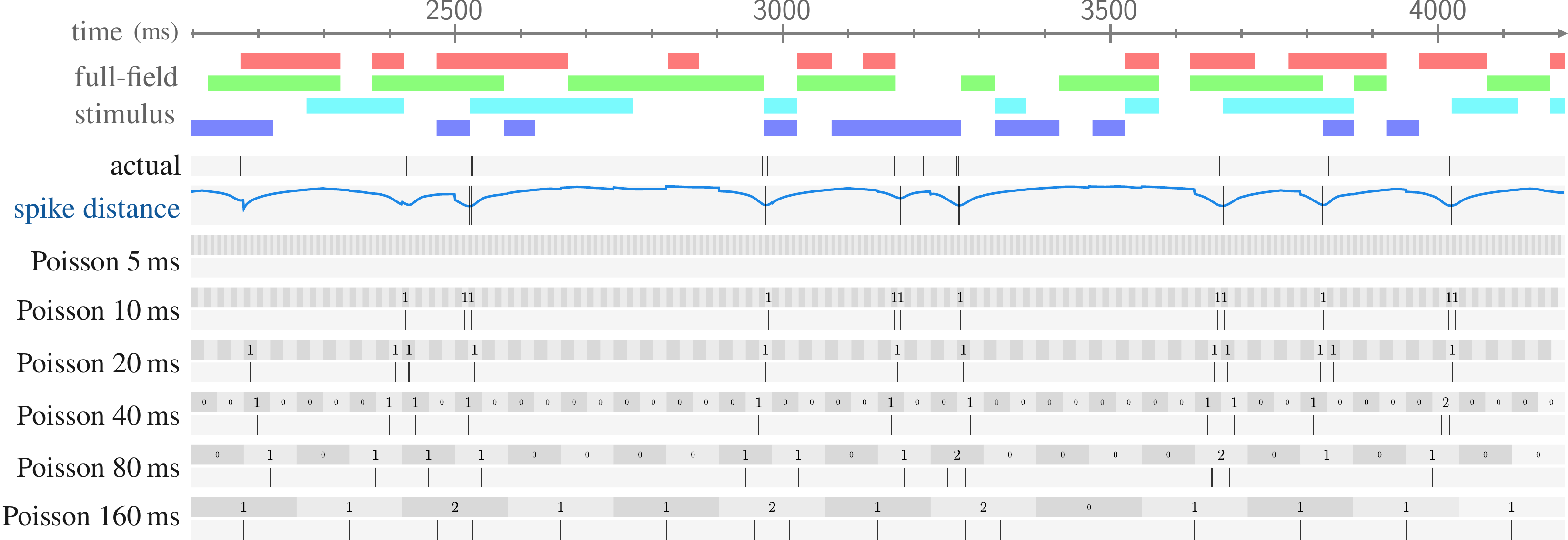}
    \captionof{figure}{\textbf{Spike prediction using spike distance.} The output of a spike distance model (depicted by the blue line) can be used to precisely predict individual spikes. Section \ref{sec:inference_dist} describes the algorithm used for inferring spikes from a spike distance array. The Poisson models have their outputs rounded and that number of spikes are tiled over the interval. With this inference strategy, a fixed summation interval introduces a trade-off: longer intervals prevent precise spike localization, while shorter intervals cause predictions to tend towards 0 for all time steps. Section \ref{sec:inference_poisson} and \ref{sec:inference_poisson2} describe other approaches for inferring spikes from outputs of Poisson models.} \label{fig:rollout}
\end{figure*}

In this work, we propose a summation-free approach. Rather than reduce a chunk of a spike train to a single number, we choose a representation that preserves its details. Inspired by the idea of a signed distance function used in geometric modelling, we propose the idea of a \emph{spike distance function} that maps points in time to the temporal distance to the nearest spike. We train neural networks to output this representation and use it to predict spike trains. A demonstration of the spike distance approach (still imprecisely defined for now) compared to the Poisson approach is shown in Figure \ref{fig:rollout}. The process of predicting spikes is shown in Figure \ref{fig:pipeline_A}. 

\begin{figure*}[t]
    \centering
    \includegraphics[width=\textwidth]{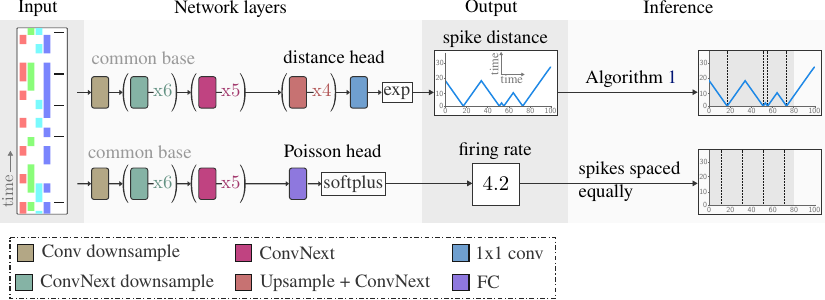}
    \captionof{figure}{\textbf{Overview} of spike prediction for both the spike distance approach and the Poisson approach. To make the two comparable, the networks used by each approach were designed with the preponderance of weights and computation carried out using the same twelve-layer base architecture. ConvNext refers to the layers introduced by \citet{liuConvNet2020s2022}. Algorithm \ref{alg:spike_inference} is described in Section \ref{sec:inference_dist}.}
    \label{fig:pipeline_A}
\end{figure*}

We inspect the effectiveness of this new approach by comparing a neural network trained using the spike distance objective with neural networks trained using a range of Poisson objectives, each using a different summation interval. Comparison is based on models' ability to infer spike trains: we predict spike trains from model outputs and compare these against ground truth spike trains using spike train similarity measures. 

Our spike prediction dataset is formed from a \ac{MEA} recording of 60 chicken \acp{RGC} responding to visual stimuli. We repeat the experiments for 113 frog \acp{RGC} from a separate recording. 

In the next section, we present related work. We come back to set the scene more concretely by elucidating the drawbacks of the Poisson approach in Section \ref{sec:poisson} and describing the spike distance function in Section \ref{sec:distance}. How to infer spike trains from model outputs is covered in Section \ref{sec:inference}. The experiments are described in Section \ref{sec:experiments} with results appearing in Section \ref{sec:results}, and a discussion follows in Section \ref{sec:discussion}. 

\section{Related Work} \label{sec:related}
For the task of spike prediction, the effectiveness despite simplicity of \acp{GLM} makes them important models to study. Using \acp{GLM} to model neurons is covered by \citet{weberCapturingDynamicalRepertoire2017}. \acp{GLM} are further relevant because their success has contributed to the popularity of the Poisson learning objective. 

\acp{GLM} are effectively single-layer neural networks, and so the neural network approaches that follow can be thought of as trading the simplicity of \acp{GLM} with the increased power of deeper neural networks. \citet{mcintoshDeepLearningModels2016} used \acp{CNN} and \acp{RNN} to predict spiking behaviour of tiger salamander \acp{RGC}. In parallel, \citet{battyMultilayerRecurrentNetwork2016} carried out similar work using \acp{RNN} to predict spiking behaviour of macaque \acp{RGC}. \citet{gogliettinoModelingResponsesMacaque2024a} used \acp{CNN} to predict responses of human \acp{RGC}. These three studies recorded the cells using \acp{MEA}. \citet{cadenaDeepConvolutionalModels2019} compared \acp{CNN} trained from scratch with a fine-tuned vision model at predicting spiking behaviour of macaque V1 neurons. This study recorded cells using penetrative probes. 
In all of these studies, models have a single output which is interpreted as a parameter of a probability distribution modelling the spike count in an interval. 

At a finer scale, spike rates are no longer a suitable abstraction, and spike prediction becomes an example of the general problem of locating points in 1 or more dimensions or locating surfaces in 2 or more dimensions and so on, of which there are other applications. Earthquake prediction faces such a problem where the task is framed as predicting the ``time to failure''. Computer vision and graphics have an analogous concept—a signed distance function—which is effectively a ``distance to surface'' and is used for implicitly modelling 2D surfaces in 3D space. It is from these contexts that we draw inspiration to propose ``distance to spike'' as an effective tool for spike prediction. 

For earthquake prediction, \citet{rouet-leducMachineLearningPredicts2017a} and \cite{wangPredictingFutureLaboratory2022} applied machine learning to predict time-to-failure for a laboratory fault system. See the review by \cite{renChapterTwoMachine2020} for broad coverage. Interestingly, these approaches stop at predicting time-to-failure signals and do not take the next step of inferring precise rupture times from these signals. The inference procedure described in Section \ref{sec:inference} will work with these signals also.

In computer vision, the notion of a signed distance function is used to implicitly define surfaces. Our work to predict spike distance is analogous to the work by \citet{zeng3DMatchLearningLocal2017a} and \citet{daiShapeCompletionUsing2017} who used neural networks to represent 3D shapes using a grid of truncated signed distance values. Later work in this area has tackled the scaling issues with 3D data, such as the work by \citet{parkDeepSDFLearningContinuous2019b}, where neural networks were given the additional job of querying the 3D representation. The directional queries of the 3D representation are conceptually similar to spike-stimulus history input that we will use to request output behaviour from our neural network model. Although not considered in this work, the creative forms of querying 3D representations, such as \citeauthor{sitzmannLightFieldNetworks2021a}'s (\citeyear{sitzmannLightFieldNetworks2021a}) use of Pl{\"u}cker coordinates, raise the question: are there other useful ways to query neural network representations for neuronal behaviour?

Instead of predicting spikes in an interval, an alternative is to predict the time until the next spike. The area of neural temporal point processes tackles this problem for arbitrary event sequences by outputting a probability distribution over the next event time, given an event history. Two reviews of this area are \citet{shchurNeuralTemporalPoint2021a} and \citet{bosserPredictiveAccuracyNeural2023}. A drawback of this approach is that inference requires a forward pass through the neural network after each spike, which becomes impractical for applications such as retinal prosthetics, as spikes of a single cell can be separated by a few milliseconds, and across cells, spikes can be recorded as coincident. Additionally, the stimulus stream strongly affects responses and quickly invalidates any prediction made beyond a short future interval, making next-spike prediction less suitable in this setting.

Outside the context of neural networks, other works have investigated non-Poisson approaches to modelling spike emission. For example, both the work of \citet{tehGaussianProcessModulated2011} and \citet{liuBayesianNonparametricNon2023} model spike emission with renewal processes modulated by Gaussian processes. These works support the idea that in the context of neural networks for spike prediction, alternatives to the incumbent Poisson approach are worth exploring.

\begin{figure*}[ht]
  \centering
  \includegraphics[width=\linewidth]{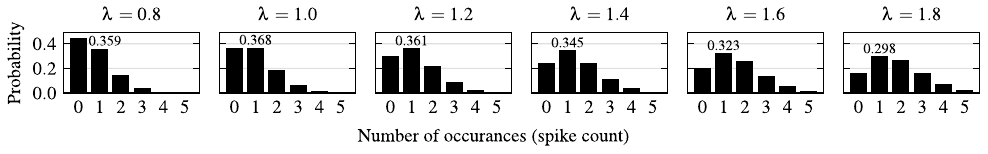}
  \caption{The probability mass function of the Poisson distribution for six values of the distribution's parameter, $\lambda$. The probability assigned to $1$ is labelled. For a single non-zero value such as $1$, no $\lambda$ allows a concentration of probability at that value. The distribution is unsuitable as part of a learning objective where a model is expected to express strong confidence in any single non-zero value.}
\label{fig:poisson_range}
\end{figure*}

\section{Issues With the Poisson Approach} \label{sec:poisson}
The summation interval used in the Poisson approach to spike prediction leads to a trade-off between temporal resolution and effective training and inference.

The goal of having predictions with high temporal resolution encourages the choice of a shorter summation interval. Recordings of chicken \acp{RGC} show bursts of spikes with spike intervals as short as \numrange[range-phrase=--]{1}{2} \si{ms} \citep{seifertBirdsMultiplexSpectral2023}; and primate, salamander, cat and rabbit \acp{RGC} have been shown to have stimulus-driven variability across trials as low as \qty{1}{\ms} \citep{uzzellPrecisionSpikeTrains2004, berryRefractorinessNeuralPrecision1998, keatPredictingEverySpike2001, berryStructurePrecisionRetinal1997}. Neurons from other brain regions also demonstrate this precision, such as neurons in the lateral geniculate nucleus of cats \citep{buttsTemporalPrecisionVisual2011, keatPredictingEverySpike2001}. Furthermore, there is evidence that this precision is linked to function: retinal ganglion cells have been shown to vary their response time to stimuli on the order of milliseconds depending on the spatial structure of stimuli \citep{gollischRapidNeuralCoding2008} and the wavelength of stimuli \citep{seifertBirdsMultiplexSpectral2023}.

This suggests that summation intervals as short as \qty{1}{ms} in duration would be useful. However, decreasing the summation interval can cause several issues. Training slows and becomes more difficult as data becomes spread over an increased number of samples, most of which are empty of spikes. If carrying out inference by maximizing probability, zero spikes can become the most compatible prediction for all time steps, as short summation intervals can lead to model outputs being consistently close to zero. These issues are ameliorated with longer intervals, but at the cost of temporal resolution.

With a long interval of say \qty{80}{ms}, \qty{80}{ms} of neural activity will be compressed into a single scalar value—a quick burst of spikes a few milliseconds apart will be indistinguishable from the same number of spikes spread over the entire interval. At training time, this results in a loss signal that has no sensitivity to how spikes fall within the interval but is overly sensitive to subtle movements of spikes crossing the interval boundary. At inference time, a model's scalar output does not indicate the positions where spikes might occur, despite it being plausible that a capable model would estimate this information in some form before arriving at a spike rate. 

The above issues will affect any setup where training and inference works with isolated bins of spike counts. Using a Poisson distribution has an additional issue in that its shape is restrictive. The distribution's single parameter is simultaneously the mean and variance of the distribution, which does not allow a model to express strong confidence in any specific non-zero spike count. For example, the largest probability assignable to the occurrence of 1 spike is $36.8\%$, which occurs when the Poisson's $\lambda$ parameter is set to $1.0$; however, at this $\lambda$, $36.8\%$ is also assigned to the occurrence of 0 spikes. Figure \ref{fig:poisson_range} demonstrates this situation by displaying the Poisson distribution for a range of $\lambda$ values between $0.8$ and $1.8$. This becomes an issue at low summation intervals as the preponderance of spike counts will be 0 or 1. 

The exponential distribution is the inter-event interval equivalent of the Poisson distribution. When modelling inter-event intervals, it is common to deviate from the exponential distribution, for example by using self-exciting Hawkes processes \citep{hawkesSpectraSelfExcitingMutually1971}. This supports the idea that it is reasonable to seek flexibility beyond what the Poisson distribution offers.

With spike prediction in mind, it would be desirable to have a loss function that changes proportionately to changes in spike position and to have an output representation that facilitates spike position inference. Furthermore, a model should not be hindered from expressing strong confidence in specific spike patterns. The spike distance function described next can be used to create a learning objective with these properties.

\section{Spike Distance} \label{sec:distance}
A \emph{Spike distance function} maps each point in time to a scalar representing the temporal distance to the nearest spike. This is either the time elapsed since the last spike or the time remaining until the next spike, whichever is less. This is an implicit representation of a set of spikes using contours. A comprehensive treatment of using implicit functions to represent points and surfaces is \citet{osherLevelSetMethods2003}.

Formally, let $S \in \mathbb{R}^N$ be a set of $N$ spike times. Let $I = [a, b]$ be an interval of $\mathbb{R}$. The spike distance function of $S$ on $I$ is the function $f_S: I \rightarrow \mathbb{R}$ defined as:
\[
  f_S(t) = \min_{s \in S} |t - s|
\]

Any interval can support a spike distance function, and this spike distance function can depend on spikes both inside and outside the interval. An example spike distance function evaluated over the interval $[\qty{0}{ms}, \qty{128}{ms}]$ for the spikes $[\qty{20}{ms}, \qty{60}{ms}, \qty{65}{ms}, \qty{86}{ms}]$ is shown in Figure \ref{fig:spikedistance}. 

\begin{figure}[H]
  \centering
  \includegraphics[width=\columnwidth]{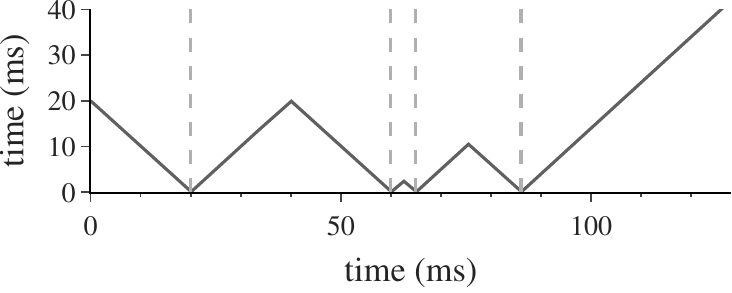}
  \caption{The spike distance function in milliseconds over the interval $[0, 128]$ for the spikes $[20, 60, 65, 86]$.}
  \label{fig:spikedistance}
\end{figure}

From the example, it can be seen that subtle or large changes to a spike train will be reflected proportionately in the spike distance function, which was the characteristic lacking in the Poisson loss signal.

In practice, we will work with discrete arrays and use a \textit{discrete spike distance function}, $\hat{f}_S : \mathbb{N} \to \mathbb{R}$, where $S$ is now an array of spike counts. We assume that the sample rate is fast enough that no two spikes are recorded in the same bin, and so $S$ can be treated as a binary array. Appendix \ref{app:discrete_distance} covers the discrete spike distance in more detail, including generalizing it to allow multiple spikes to be recorded in a single sample. Two other considerations—the choice of a maximum distance and the choice of the evaluation interval—are discussed in Appendix \ref{app:spike_distance}.

\section{Spike Train Inference}  \label{sec:inference}
For the spike prediction task, we will infer a spike train from a model's outputs. For both the spike distance model and the Poisson models this will be done autoregressively: concatenating model outputs and using previously predicted spikes as input to subsequent forward passes. The two approaches differ in how they carry out a single step—each is covered below.

\subsection{Spikes From a Spike Distance Array} \label{sec:inference_dist}
We describe spike inference from a spike distance array from the perspective of energy minimization. An array of spike distance values outputted by a model can be seen as parameterizing some function that assigns a scalar energy value to any set of spikes. Various functions are possible. We use the L2 norm between the model output and the discrete spike distance function of a candidate spike sequence. Let the model output $M = (m_i)_{i=0}^{L-1}$ be a sequence of reals of length $L$, and let the binary sequence $S=(s_i)_{i=0}^{L-1}$ be a candidate spike count sequence of the same length. For example, with $L = 5$, the model output could be $M = [1.1, 0.5, 0.9, 1.1, 0.6]$ and a candidate spike sequence count be $S = [0, 1, 0, 0, 1]$.  The energy assigned to $S$ is given by:

\begin{equation}
  E_M(S) = \sum_{i=0}^{L-1} (\hat{f}_S(i) - m_i)^2 
  \label{eq:energy}
\end{equation}

\begin{figure}[h]
  \centering
  \includegraphics[width=\columnwidth]{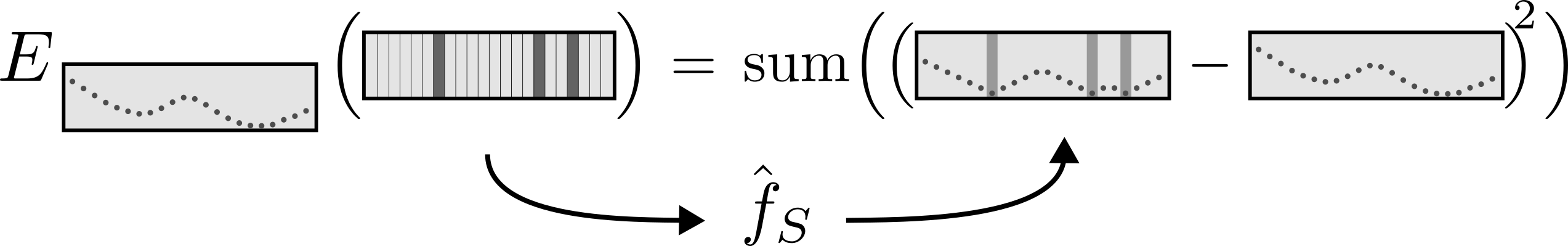}
  \caption{Visual form of Equation \ref{eq:energy}. The discrete spike distance, $\hat{f}_S$, is described further in  Appendix \ref{app:discrete_distance}.}
\label{fig:energy_eq}
\end{figure}

The energy is a measure of incompatibility between the candidate spikes and the model output (here we are following the convention that low energy values correspond to greater compatibility). Spike train inference amounts to finding the spike train with the lowest energy. 

This optimization problem can be solved exactly by brute-force search. Dynamic programming can also be used, as the problem can be decomposed into left and right sub-problems. Both of these approaches were found to be too slow in practice. We propose an inexact solution, Algorithm \ref{alg:spike_inference}, that begins by predicting a spike in every time step then iteratively refines the prediction by removing spikes.

\begin{algorithm}[H]
\caption{Spike inference algorithm: predict a spike in every time step then iteratively remove spikes—remove if the L2 norm between the candidate and target distance arrays is lower when the spike is not present. Iterate spikes in order of the estimated effect on the error. The function $\textsc{SpikeDist}$ converts a set of spikes into an array of spike distances (see Appendix \ref{app:discrete_distance} for details).}

\label{alg:spike_inference}

\begin{algorithmic}[1] %
\REQUIRE{$w$ is the target spike distance (model output) of length $L$ and $s_0$ is the most recent ``past'' spike.}
\FUNCTION{\textsc{SpikeInference}($w, s_0$)} 
\STATE $S \gets \text{indices of } w \;\; (\text{i.e. } [0, 1, ..., L-1]) $
    \STATE $\textit{scores} \gets w$
    \REPEAT
    \STATE $S_{\textit{old}} \gets S$
    \FORALL {$s \in S$ sorted descending by $\textit{scores}[s]$}
      \STATE $err \gets \textsc{L2Norm}(w, \textsc{SpikeDist}(S, s_0)\,)$
          \STATE $S' \gets S \text{ with } s \text{ removed}$
          \STATE $err' \gets \textsc{L2Norm}(w, \textsc{SpikeDist}(S', s_0)\,)$
          \STATE $\delta \gets err - err'$
          \STATE $\textit{scores}[s] \gets \delta$
          \IF{$\delta > 0$}
              \STATE $S \gets S'$
          \ENDIF
      \ENDFOR
      \UNTIL{$S = S_{\textit{old}}$}
  \STATE \textbf{return} $S$
\ENDFUNCTION
\end{algorithmic}
\end{algorithm}

In terms of time complexity, applying this algorithm autoregressively to infer a spike train $L\,\text{bins}$ long, broken into $N$ steps of $l\,\text{bins}$ uses $\sim N l^2 \log l$ compares in the worst case—this corresponds to an $l \log l$ cost of sorting that may be repeated $l$ times each inference step, of which there are $N$. Except in pathological cases, the outer loop of \textsc{SpikeInference} requires very few iterations—for the combined 90 minutes of spike train inference carried out on the chicken test set, no inference step required more than two iterations. Thus, with comparisons as a cost model, the algorithm's average cost is $\sim N l \log l$: growing linearly with the number of inference steps ($N$) and linearithmic with the resolution ($l$) of a single inference step. 
Figure \ref{fig:inference_perf} in Appendix \ref{app:comp_resources} shows inference running times on the authors' hardware. The algorithm uses additional memory that grows linearly with $l$.

\subsection{Spike Counts From a Poisson Distribution} \label{sec:inference_poisson}
A Poisson distribution has a single parameter $\lambda$. A model trained with a Poisson learning objective outputs a single scalar value, $y$, which is used to parameterize a Poisson distribution by setting $\lambda = y$. From such a parameterized Poisson distribution, we consider three ways to infer spike counts: maximizing probability, sampling, and approximating the mean.

\textbf{Maximizing probability} involves selecting the spike count with maximum probability under the Poisson distribution with parameter $\lambda=y$. By nature of the Poisson distribution, this value—the mode of the distribution—is $\lfloor y \rfloor$.

\textbf{Sampling} a spike count from the Poisson distribution with parameter $\lambda=y$ is another approach to inferring a spike count. Instead of selecting the most likely spike count, the distribution is sampled to produce a random spike count. 

\textbf{Approximating the mean} is an alternative whereby the model's output is rounded to produce a spike count prediction. This approach is a less principled approach; however, empirically, it outperforms sampling or using the mode. Indeed, the overall best-scoring spike trains using Poisson models were produced with this approach at a summation interval of \qty{80}{\ms}. It is hypothesized but not tested that the reason for the rounded mean being more effective than the mode is that the mean is always larger and so will take the value zero less often in the case where the Poisson distribution skews towards zero. 

The main results are reported using both sampling and the rounded mean strategy. Inference by maximizing probability (using the mode) performs significantly worse than using the rounded mean and is reported in Appendix \ref{app:floor}.

\subsection{Spikes From a Spike Count} \label{sec:inference_poisson2}
The Poisson distribution gives no information other than the spike rate, so a reasonable strategy for spike placement must be chosen. In this work, given an inferred spike count $n$, we tile $n$ spikes uniformly over the interval. In experiments not reported here, stochastic spike placement was also tested—by sampling $n$ times from a uniform distribution; however, this approach performed notably worse than uniform tiling. 

\begin{figure*}[ht]
  \centering
  \includegraphics[width=\linewidth]{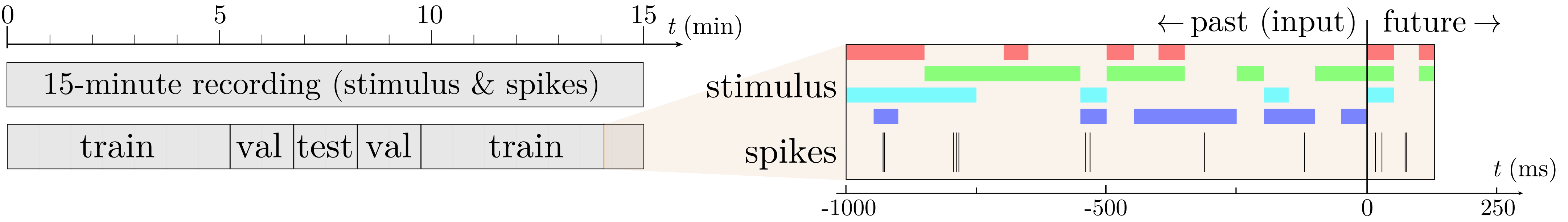}
  \caption{\textbf{Left:} A 15-minute retina recording, split into training, validation and test subsets. \textbf{Right:} A single \textasciitilde 1.1-second snippet taken from the training subset. ${\small t=0}$ separates the snippet into the ``past'' data, available to the model, and the ``future'' data, for which models will be expected to make predictions. In our experiments, the length of the past data is always 1 second (992 samples); the length of the future data depends on the model; for example, the Poisson-\qty{5}{ms} model with a 5-sample summation interval will use 5 samples (\textasciitilde\qty{5}{\ms}).}
\label{fig:data}
\end{figure*}

\section{Experiment Settings} \label{sec:experiments}
A spike distance model was compared to a family of 6 Poisson-based models, each differing by the length of their summation interval: \qty{5}{\ms}, \qty{10}{\ms}, \qty{20}{\ms}, \qty{40}{\ms}, \qty{80}{\ms} and \qty{160}{\ms}. Below we describe the dataset, models, training procedure and evaluation procedure.

\subsection{Dataset} \label{sec:dataset}
Model input was a (5, 992)--shaped array corresponding to 1 second of history: a 4-channel stimulus and a single spike train, both sampled at \qty{992}{Hz}. The data came from a 15-minute recording of chicken \acp{RGC} exposed to full-field colour noise, recorded by \citet{seifertBirdsMultiplexSpectral2023}. The recording contains 154 cells. Models were trained separately for each cell, so each cell can be considered a separate 15-minute dataset. The 15 minutes was split according to the ratio (7,2,1) into training, validation and test sets. Appendix \ref{app:recording} describes the recording in more detail. 

Our computational budget restricted how many cells we could work with—training time was considered too long for more than 60 cells (60 cells requires 500 hours of training, as described in Appendix \ref{app:comp_resources}). A principled way to choose cells is to focus on cells with more spikes. The evaluation metrics that will be introduced in Section \ref{sec:evaluation} become less informative at low spike counts—the spike train metrics need spikes to operate on, and at low spike counts, it is difficult to distinguish models from a naive model that outputs an empty spike train. We consider only cells that contain an average spike rate greater than 0.75 spikes per second in the training set. This threshold resulted in 60 of 154 cells being used. Cells that meet this criterion will typically have more than 65 spikes in the 90-second test set (the exact number depends on their spike rate in the test set). Appendix \ref{app:dataset} describes the construction of the dataset in more detail.

\subsection{Models} \label{sec:models}
For each of the 60 cells, 7 models were trained: 1 spike distance model and 6 Poisson models. Each Poisson model used a different summation interval length: 5, 10, 20, 40, 80 and 160 bins. As the duration of 1 bin (\qty{1.0008}{\ms}) is approximately \qty{1}{\ms}, these models are referred to as Poisson-\qty{5}{\ms}, Poisson-\qty{10}{\ms}, and so on.

A single architecture is not used for both the spike distance objective and the Poisson objective, as the former requires output in the form of an array and the latter a scalar; instead, two neural networks are used. These were designed towards maximizing our ability to compare the two learning objectives. An overview of the architectures is shown in Figure \ref{fig:pipeline_A}. The two architectures were designed around a shared base architecture that accounts for the preponderance of parameters and compute. A \ac{CNN} using ConvNext blocks described by \citet{liuConvNet2020s2022} forms the shared base. The Poisson architecture then ends in a fully-connected layer, whereas the spike distance architecture's head consists of 4 upsampling ConvNext blocks. There is no weight sharing employed for training on multiple cells—the models are trained end-to-end for each cell individually. The alternative of training on all cells at once by adding per-cell modulation to the network was avoided, as the outputs of the two architectures are very different in nature, and any single modulation scheme may favour one over the other. Further details on the models are described in Appendix \ref{app:models}.

One concern addressed is layer count. The spike distance architecture has additional layers in its head, and we wished to rule out this difference leading to a performance advantage. This concern was ameliorated by optimizing the number of layers in the shared base architecture to maximize Poisson performance. This architecture search over layer count is described in Appendix \ref{app:arch_search}. 

\begin{figure*}
\begin{subfigure}[t]{\columnwidth}
  \centering
  \includegraphics[width=0.93\linewidth]{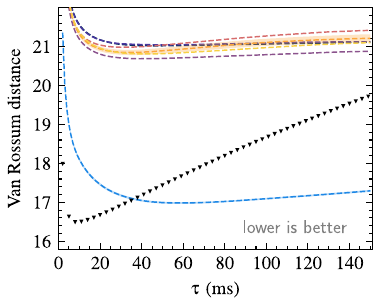}
  \bigbreak
  \includegraphics[width=0.93\linewidth]{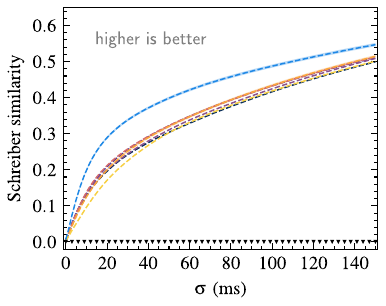}
  \bigbreak
  \includegraphics[width=0.93\linewidth]{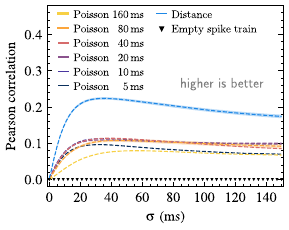}
  \caption{Poisson models use inference via sampling}
  \label{fig:sampling_xsigma}
\end{subfigure}
\hfill
\begin{subfigure}[t]{\columnwidth}
  \centering
  \includegraphics[width=0.93\columnwidth]{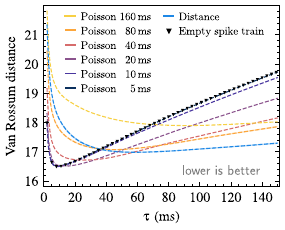}
  \bigbreak
  \includegraphics[width=0.93\columnwidth]{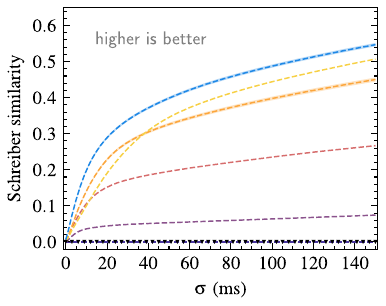}
  \bigbreak
  \includegraphics[width=0.93\columnwidth]{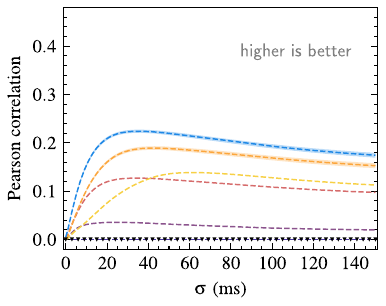}
  \caption{Poisson models use inference via the rounded mean}
\label{fig:round_xsigma}
\end{subfigure}
\caption{Performance comparison between the spike distance model and the Poisson models at spike prediction. For the Poisson models, two spike count inference strategies are considered: \textbf{(left)} sampling and \textbf{(right)} using the rounded mean. Three metrics are considered: \textbf{(top)} Van Rossum distance, \textbf{(middle)} Schreiber similarity and \textbf{(bottom)} Pearson correlation between the ground truth and output spike trains, reported as interquartile mean over the 60 chicken cells. The metrics are evaluated for a range of their smoothing parameters, ($\tau$ and $\sigma$), from 0 to 150. 95\% confidence intervals are included for the Poisson-\qty{80}{\ms} and spike distance models. Scores for the zero spike spike train are included for comparison, highlighting that the metrics become less effective at comparing models at low smoothing levels—in particular, the Van Rossum distance eventually prefers the zero spike train for sufficiently low $\tau$ values. See Figure \ref{fig:mxm} in Appendix \ref{app:mxm} for a visualization of the variability across cells at a specific smoothing value, $\tau=\sigma=60$.}
\label{fig:xsigma}
\end{figure*}

\subsection{Training}
All models were trained for 80 epochs using the AdamW optimizer with the 1-cycle learning rate policy \citep{smithCyclicalLearningRates2017}. Loss for the spike distance model was mean squared error between the model output and the log of the target spike distance. Loss for 
Poisson models was negative log-likelihood with respect to the target Poisson distribution. The $\lambda$ parameter of the target distribution is the count of spikes within the summation interval. Final models were chosen based on the lowest validation loss across the 80 epochs. Training hyperparameters are described further in Appendix \ref{app:hyperparameters}. Total training times are summarized in Appendix \ref{app:comp_resources}.

\subsection{Evaluation} \label{sec:evaluation}
Once trained on a cell, a model was used to predict a spike train using the 90-second test data. Each inferred spike train was compared to the ground truth spike train using three metrics for spike train similarity/dissimilarity: Van Rossum distance \citep{vanrossumNovelSpikeDistance2001}, Schreiber similarity \citep{schreiberNewCorrelationbasedMeasure2003} and Pearson correlation. Appendix \ref{app:metrics} covers these metrics in more detail. Schreiber similarity and Van Rossum distance were chosen because they are frequently used and easy to explain. Pearson correlation is reported so as to have a common metric with the work of \citet{mcintoshDeepLearningModels2016} and \citet{gogliettinoModelingResponsesMacaque2024a}. We refer to \citet{paivaComparisonBinlessSpike2010} and \citet{sihnSpikeTrainDistance2019} for a comparison of different spike train measures. 

The spike train metrics are each parameterized by a smoothing parameter. For Van Rossum distance and Schreiber similarity, their parameters—$\tau$ and $\sigma$ respectively—are integral to their definition. We augment the Pearson correlation by prepending a smoothing step copied from the Schreiber similarity. The comparison of spike trains is sensitive to the degree of smoothing, and different models perform best at different smoothing ``sweet spots''. There is a scale from instantaneous spikes to average spike rate. No single point on this scale can capture all that matters about a neuron's response, so it is important to investigate spiking behaviour over a range of timescales. We evaluate the three metrics at a range of smoothing values, from 0 to 150.

\section{Results} \label{sec:results}
Training and evaluating a model 60 times, once for each cell, is considered a run. 11 runs were performed for both the spike distance model and the Poisson-\qty{80}{\ms} model (the most competitive model). All other models were evaluated for a single run. We follow the recommendation of \citet{agarwalDeepReinforcementLearning2021} and report aggregate performance for a model using \ac{IQM} across cells and runs; and for the two models with multiple runs, uncertainty is estimated with stratified bootstrap confidence intervals. 

Figure \ref{fig:xsigma} records the three metrics for all models. Both sampling (Figure \ref{fig:sampling_xsigma}) and rounded mean (Figure \ref{fig:round_xsigma}) are used for inference with the Poisson models. An additional perspective on the variability across cells is provided by Figure \ref{fig:mxm} (Appendix \ref{app:mxm}) which zooms in on a single smoothing parameter. The results when the Poisson models use inference via the mode are presented in Figure \ref{fig:floor_xsigma} (Appendix \ref{app:floor}). For a qualitative comparison of inference strategies for the Poisson models, the snippet from Figure \ref{fig:rollout} is repeated in Appendix \ref{app:rollout} for inference using sampling (Figure \ref{fig:rollout_sampling}) and the distribution mode (Figure \ref{fig:rollout_floor}).

The experiment was repeated for a recording of frog \acp{RGC}, and the results for this data are presented in Appendix \ref{app:xenopus}.

\clearpage

\section{Discussion} \label{sec:discussion}
An effective model should perform well over a wide range of the metrics' smoothing parameters, and in this regard, the spike distance model outperforms the Poisson models in all three metrics. 

For the Poisson models, the spike inference strategy has a major effect on the quality of the spike trains produced. When using inference via sampling (Figure \ref{fig:sampling_xsigma}), all Poisson models perform similarly but are outperformed by the spike distance model, and considerably so in terms of Van Rossum distance. Inference via the rounded mean allows a more deterministic placement of spikes; however, as the summation interval is reduced, models begin to approach the performance of the zero spike train model due to model outputs being rounded to zero. Eventually, the model with the shortest summation interval, Poisson \qty{5}{\ms}, has performance indistinguishable from that of the zero spike train. When inference uses the distribution's mode, (Figure \ref{fig:floor_xsigma} in Appendix \ref{app:floor}) this effect is exacerbated, and both Poisson \qty{20}{\ms} and Poisson \qty{10}{\ms} also have performance indistinguishable from that of the zero spike train.

The Poisson \qty{80}{\ms} model using inference via rounded mean is the overall best-performing Poisson model. It can be understood to sit in a sweet spot where the summation interval is long enough to be less affected by the rounding to zero effect but short enough to have good temporal resolution. This result suggests that the less principled approach of inferring spike counts by rounding the model output should be considered when choosing a spike count inference strategy. If one is unable to experiment with various summation intervals, using sampling can give a higher lower-bound performance across summation intervals at the cost of incurring a lower peak performance. 

The evaluation settings that least favour the distance model are the Van Rossum distances at low smoothing levels, where the Van Rossum distance prefers Poisson models with short summation intervals using the mode or rounded mean for inference. This is explained by these models predicting very few spikes and having a performance similar to the empty spike train model. For Van Rossum distance, there will always be a $\tau$ value below which predicting zero spikes performs best. At low smoothing values, inexact spikes are penalized twice—once for missing the correct spike and once for a spike at an incorrect time. \citet[p.~408]{paivaComparisonBinlessSpike2010} describes this as Van Rossum distance acting as a co-incidence detector at low $\tau$ values. When comparing models via the Van Rossum distance, the $\tau$ values where the empty spike train performs best are less informative, and above these $\tau$ values, it is the spike distance model that is most competitive.

\section{Limitations and Improvements}
The results supporting the use of the spike distance objective are limited by the scope of the experiment. The two learning objectives were compared in the context of a neural network model with mostly convolutional layers, and how the properties of this architecture, such as the translation invariance of the convolutional layers, support or undermine each learning objective is an interesting question that is not investigated in this work. How the differences between the model outputs interact with hyperparameters or affect regularization are other considerations that are not investigated.

A strength of the Poisson objective that remains undisputed is interpretability—GLMs trained with a Poisson objective can be easily probed to obtain interpretable stimulus filters \citep{weberCapturingDynamicalRepertoire2017}. It is unclear whether models with a spike distance output would be as useful in this regard. It is only in the context of predicting precise spike times that this work argues for the use of the spike distance objective.

There are opportunities to increase the competitiveness of the spike distance approach. The model architecture can be optimized with spike distances in mind (the architecture used in this work was optimized for the performance of the Poisson models). As some aspects of generating spike distance arrays do not need real spike data to be learned, pretraining the network using synthetic data may be effective. As for the real data—it can be better utilized by training a single model on multiple cells, which was avoided in this work in order to improve our ability to compare models. Finally, the space of algorithms for inferring spikes from spike distance arrays has yet to be explored deeply. Algorithm \ref{alg:spike_inference} demonstrates effectiveness in practice; however, it is an approximate solution, and to what extent an optimal solution improves spike prediction is an interesting question. 

It is also interesting to ask: what other representations can create effective learning objectives? We experimented with other representations such as a Van Rossum-like representation and a representation that tracked separately the time until and time since a spike. With these representations, we observed model outputs that were more difficult to infer spikes with, and we did not pursue these past early experiments. Such representations are worth further investigation.

\section{Conclusion} \label{sec:conclusion}
We propose using a ``distance to spike'' representation of a spike train which we call a \textit{spike distance function}. We show that neural networks can be trained to approximate spike distance functions and that these outputs can be used to predict spikes. For the task of spike train prediction, we show that models trained to approximate a spike distance function outperform models trained using the standard Poisson objective.

\section*{Impact Statement}
This paper presents work whose goal is to advance the field of Machine Learning. There are many potential societal consequences of our work, none which we feel must be specifically highlighted here.

\section*{Acknowledgements}
KD would like to acknowledge scholarship support from the Leverhulme Trust (DS-2020-065). CY acknowledges funding from the European Union’s Horizon 2020 research and innovation programme under the Marie Skłodowska-Curie grant agreement no: 101026409. TP acknowledges funding from the Wellcome Trust (Investigator Award in Science 220277/Z20/Z), the European Research Council (ERC-StG ``NeuroVisEco'' 677687), URKI (BBSRC, BB/R014817/1, BB/W013509/1 and BB/X020053/1), the Leverhulme Trust (PLP-2017-005, RPG-2021-026 and RPG-2-23-042) and the Lister Institute for Preventive Medicine. This research was funded in whole, or in part, by the Wellcome Trust [220277/Z20/Z]. For the purpose of open access, the author has applied a CC BY public copyright licence to any Author Accepted Manuscript version arising from this submission.

\bibliographystyle{plainnat}
\bibliography{main.bib}

\vfill
 
\pagebreak
\clearpage 

\appendix

\section{Appendix}
Appendices \ref{app:mxm}-\ref{app:xenopus} present supplementary results. Appendix \ref{app:metrics} describes the evaluation metrics in more detail. Appendix \ref{app:evaluation_methodology} argues for the suitability of our evaluation methodology. Details of the MEA recording and how it is processed to form a dataset are presented in Appendices \ref{app:recording} and \ref{app:dataset}. Appendices \ref{app:models}, \ref{app:arch_search} and \ref{app:hyperparameters} cover the models used. Computational resources used are listed in Appendix \ref{app:comp_resources}. Appendices \ref{app:spike_distance} and \ref{app:discrete_distance} go into more detail on working with spike distances.

\begin{figure*}
  \begin{subfigure}[t]{\columnwidth}
      \centering
       \includegraphics[width=0.82\columnwidth]{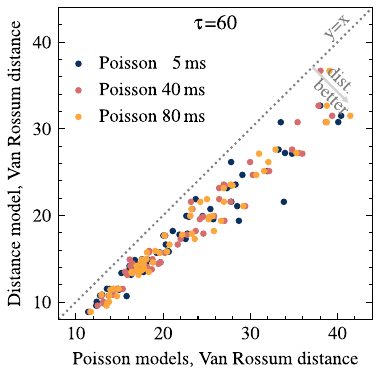}
       \bigbreak
       \includegraphics[width=0.82\columnwidth]{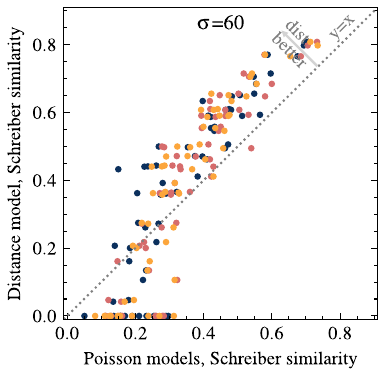}
       \bigbreak
       \includegraphics[width=0.82\columnwidth]{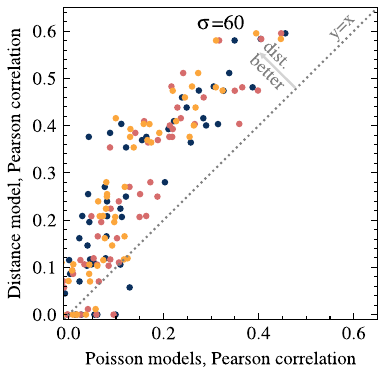}
    \caption{Poisson models use inference via sampling}
    \label{fig:sampling_mxm}
    \end{subfigure}
\hfill
\begin{subfigure}[t]{\columnwidth}
  \centering
  \includegraphics[width=0.82\columnwidth]{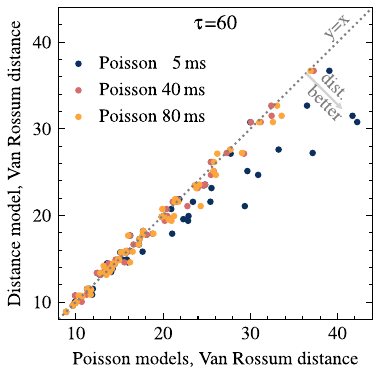}
  \bigbreak
  \includegraphics[width=0.82\columnwidth]{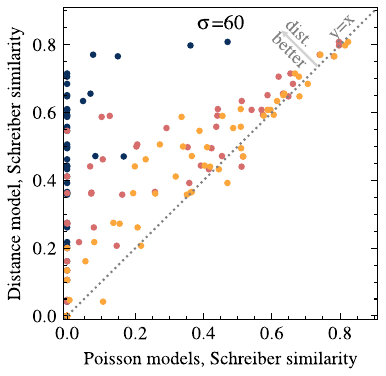}
  \bigbreak
  \includegraphics[width=0.82\columnwidth]{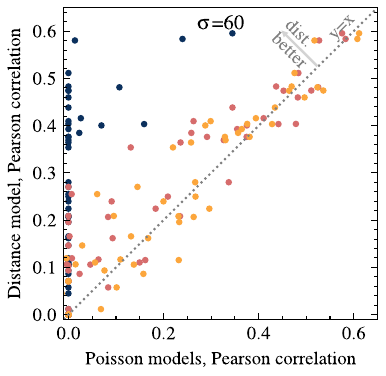}
  \caption{Poisson models use inference via the rounded mean}
  \label{fig:round_mxm}
\end{subfigure}
\caption{Zoomed in to a single smoothing value from Figure \ref{fig:xsigma}, the spike distance model is compared to 3 of the Poisson models for all 60 cells. The left and right columns of this figure correspond to the left and right columns of Figure \ref{fig:xsigma}: \textbf{(left)} Poisson models use sampling for inference, \textbf{(right)} Poisson models use the rounded mean for inference. All three metrics are presented: \textbf{(top)} Van Rossum distance at $\tau=\qty{60}{ms}$, \textbf{(middle)} Schreiber similarity at $\sigma = \qty{60}{ms}$, and \textbf{(bottom)} Pearson correlation at $\sigma = \qty{60}{ms}$. 
     }
\label{fig:mxm}
\end{figure*}

\subsection{Variability Across Cells} \label{app:mxm}
Figure \ref{fig:mxm} explores the variability of the data by zooming in on a single smoothing parameter value (the x-axis value $\tau = \sigma = 60$) of the results in Figure \ref{fig:xsigma}. The purpose of this figure is to give readers a better sense of the spread of the data that is being condensed in Figure \ref{fig:xsigma}.

\begin{figure}[h]
  \centering
  \includegraphics[width=0.92\columnwidth]{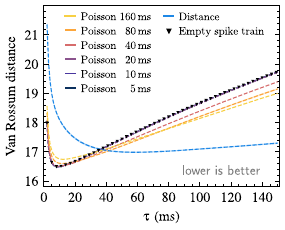}
  \bigbreak
  \includegraphics[width=0.92\columnwidth]{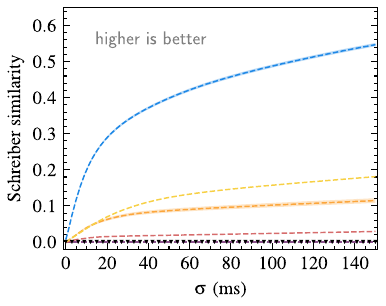}
  \bigbreak
  \includegraphics[width=0.92\columnwidth]{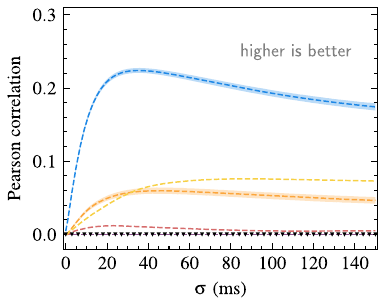}
  \caption{Model comparison when the mode of the Poisson distribution is used for spike prediction with the Poisson models. Van Rossum distance \textbf{(top)}, Schreiber similarity \textbf{(middle)} and Pearson correlation \textbf{(bottom)} between the ground truth and output spike trains, reported as interquartile mean over the 60 chicken cells. All metrics are evaluated for a range of smoothing parameters. 95\% confidence intervals are included for the Poisson-\qty{80}{\ms} and spike distance models.}
      \label{fig:floor_xsigma}
  \centering
\end{figure}

\subsection{Inference by Maximizing Probability} \label{app:floor}
During inference, using the mode of the Poisson distribution (the floor of a Poisson model's output) is a principled approach to selecting the spike count with the highest probability. Using this strategy performs significantly worse than using the rounded mean of the distribution (the rounded model's output). Figure \ref{fig:floor_xsigma} complements Figure \ref{fig:sampling_xsigma} and Figure \ref{fig:round_xsigma} by reporting the performance of the Poisson models when using the mode for inference.

\subsection{Input-Output Snippets for Inference Using Sampling and Using the Mode} \label{app:rollout}
An input-output snippet comparing the spike distance model to the Poisson models was shown in Figure \ref{fig:rollout}. The Poisson models in this figure used the rounded model outputs for spike count inference. We repeat this figure but with Poisson models using inference via sampling (Figure \ref{fig:rollout_sampling}) and inference via the mode (Figure \ref{fig:rollout_floor}).

\begin{figure*}[h]
    \centering
    \includegraphics[width=\textwidth]{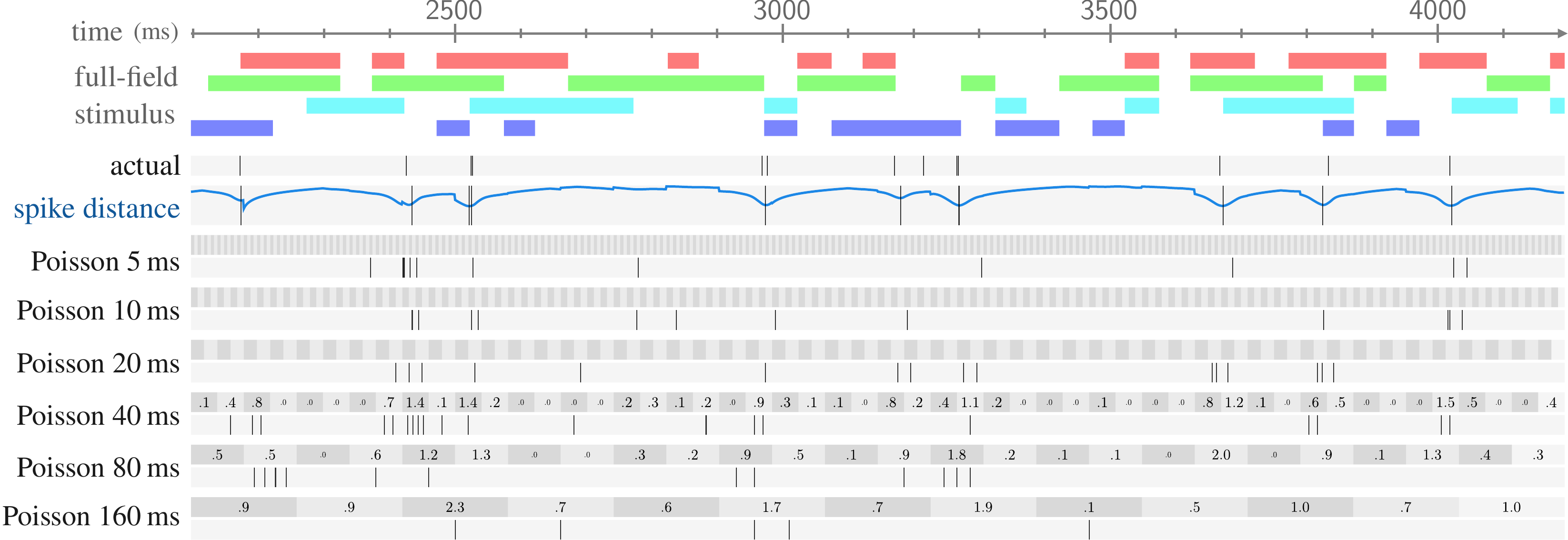}
    \captionof{figure}{\textbf{Inference via sampling, for Poisson models.} 
    The same input-output snippet from Figure \ref{fig:rollout} is repeated here, but with spike count inference via sampling being used for the Poisson models. The boxed numbers represent the neural network outputs, rounded to two decimal places (rounded to fit in the figure). When using sampling, a spike count quite different from a model's output can become the prediction. This is beneficial for short summation intervals, where outputs might otherwise be rounded or floored to zero, but with longer intervals, the model's output may be an accurate prediction that then gets perturbed by the sampling process. Note: due to limited space in the figure, the models with summation interval below \qty{40}{ms} have model output values omitted.} \label{fig:rollout_sampling}
\end{figure*}%
\begin{figure*}[h]
    \centering
    \includegraphics[width=\textwidth]{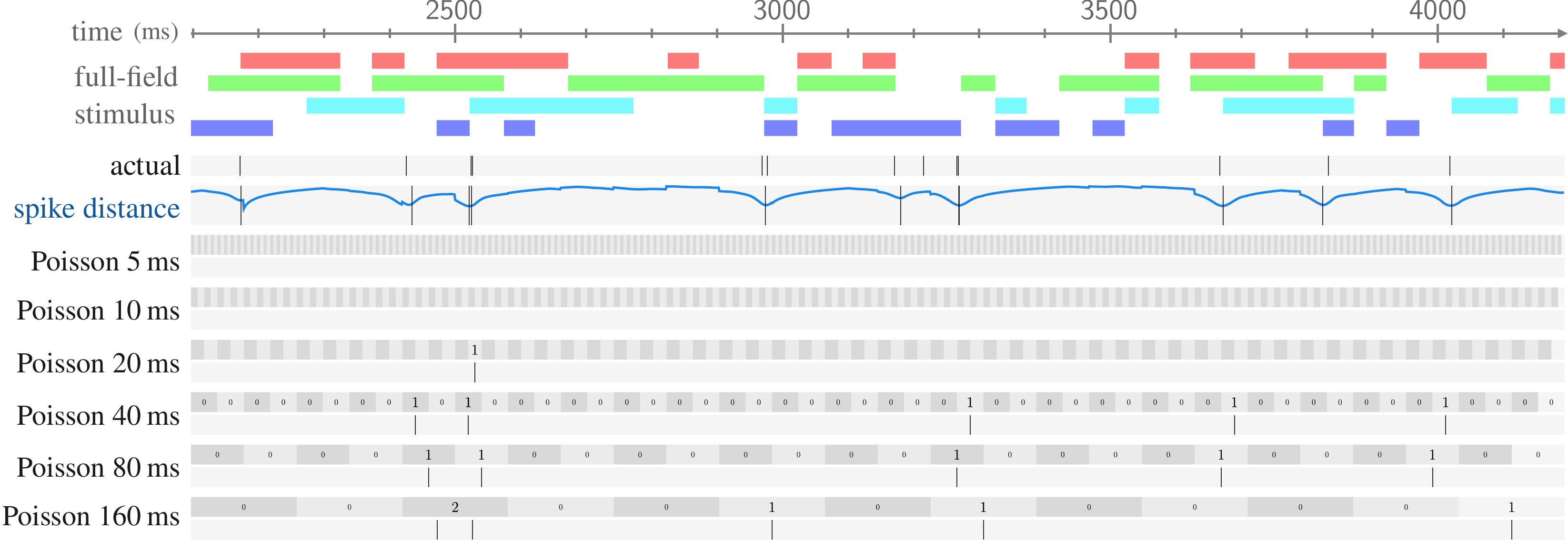}
    \captionof{figure}{\textbf{Inference via the distribution mode, for Poisson models.} 
    The same input-output snippet from Figure \ref{fig:rollout} (and Figure \ref{fig:rollout_sampling}) is repeated here, but with Poisson models using the distribution's mode for spike count inference. The boxed numbers represent the floored neural network outputs. The mode of the Poisson distribution is the floor of the model's output. For a preponderance of outputs, the model's output is floored to zero.
  Note: due to limited space in the figure, the models with summation interval below \qty{20}{ms} have their model output values omitted; predicted spikes are still shown, although there are no predicted spikes to show for the Poisson \qty{10}{ms} and Poisson \qty{5}{ms} models.} \label{fig:rollout_floor}
\end{figure*}%

\subsection{Xenopus (Frog) \acp{RGC}} \label{app:xenopus}
Xenopus (frog) \acp{RGC} were recorded under similar conditions to those of the chicken \acp{RGC}. The same \ac{MEA} device and recording process was used. One difference is that, unlike the chicken dataset for which \citet{seifertBirdsMultiplexSpectral2023} applied a quality criterion to filter out poorly responding cells, no such filtering was applied to the frog data. The same average spike rate threshold of 0.75 spikes per second in the training set was applied, and this filtering resulted in 113 cells. Figure \ref{fig:frog_xsigma} shows the Van Rossum distance, Schreiber similarity and Pearson correlation for 113 recorded frog \acp{RGC}. All models were trained for a single run only (no confidence intervals). The gap in performance between the spike distance model and the Poisson models is narrower for the frog \acp{RGC} in comparison to the chicken \acp{RGC} (Figure \ref{fig:xsigma}); however, the spike distance model is still the most consistently competitive model across the three metrics.

\clearpage

\begin{figure*}
\begin{subfigure}[t]{\columnwidth}
  \centering
  \includegraphics[width=0.93\linewidth]{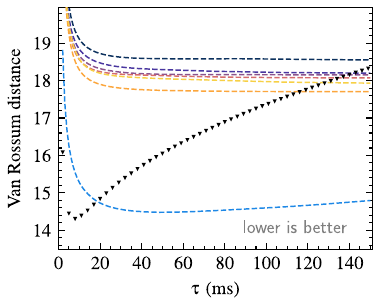}
  \bigbreak
  \includegraphics[width=0.93\linewidth]{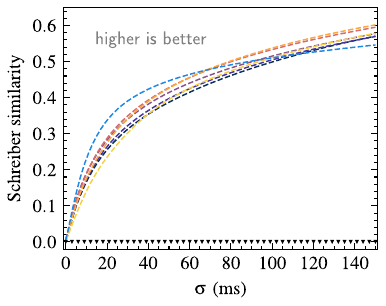}
  \bigbreak
  \includegraphics[width=0.93\linewidth]{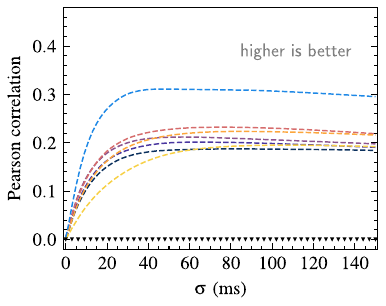}
  \caption{Poisson models use inference via sampling}
  \label{fig:frog_sampling_xsigma}
\end{subfigure}
\hfill
\begin{subfigure}[t]{\columnwidth}
  \centering
  \includegraphics[width=0.93\columnwidth]{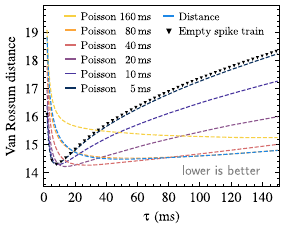}
  \bigbreak
  \includegraphics[width=0.93\columnwidth]{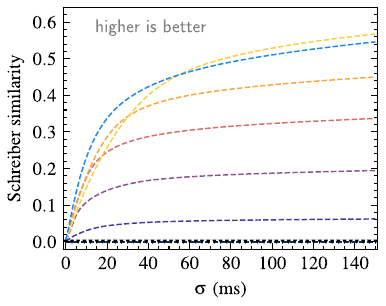}
  \bigbreak
  \includegraphics[width=0.93\columnwidth]{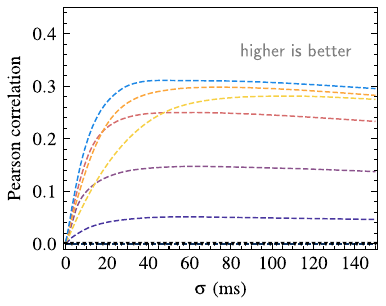}
  \caption{Poisson models use inference via the rounded mean}
\label{fig:frog_round_xsigma}
\end{subfigure}
\caption{Frog \acp{RGC}. 
Performance comparison between the spike distance model and the Poisson models at spike prediction. For the Poisson models, two spike count inference strategies are considered: \textbf{(left)} sampling and \textbf{(right)} using the rounded mean. Three metrics are considered: \textbf{(top)} Van Rossum distance, \textbf{(middle)} Schreiber similarity and \textbf{(bottom)} Pearson correlation between the ground truth and output spike trains, reported as interquartile mean over the 113 frog \acp{RGC}. All models are trained for a single run (no confidence intervals). Metrics are evaluated for a range of their smoothing parameters, ($\tau$ and $\sigma$), from 0 to 150. Scores for the zero spike train are included for comparison, highlighting that the metrics become less effective at comparing models at low smoothing levels—in particular, the Van Rossum distance eventually prefers the zero spike train for sufficiently low $\tau$ values.}
\label{fig:frog_xsigma}
\end{figure*}

\clearpage 

\begin{figure*}[ht]
    \centering
    \includegraphics[width=\textwidth]{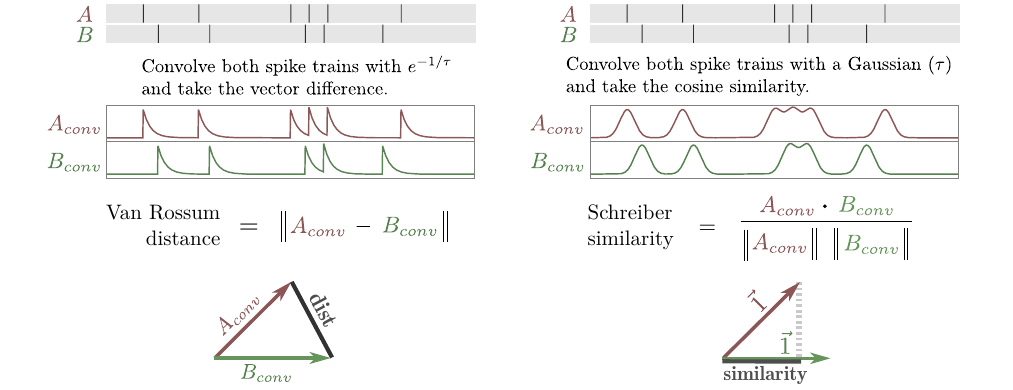}
    \captionof{figure}{ 
      Calculation of Van Rossum distance (\textbf{left}) and Schreiber similarity (\textbf{right}) from two input spike trains. Both metrics employ distinct smoothing kernels to produce an intermediate vector. Van Rossum distance can be thought of as measuring the Euclidean distance between the vectors, while the Schreiber similarity assesses the angle between them.}
    \label{fig:metrics}
\end{figure*}

\subsection{Evaluation Metrics} \label{app:metrics}
This section describes the three spike train metrics used for evaluation. Figure \ref{fig:metrics} describes Van Rossum distance and Schreiber similarity. Pearson correlation is the well-known normalized covariance. 

Choosing a spike train metric remains an unsatisfying activity as there is not yet any strong basis on which the characteristics of different metrics can be definitively ranked. As mentioned in Section \ref{sec:evaluation}, Schreiber similarity and Van Rossum distance were chosen as they are frequently used and easy to explain. Pearson correlation is reported so as to have a common metric with the work of \citet{mcintoshDeepLearningModels2016} and \citet{gogliettinoModelingResponsesMacaque2024a}. We refer to \citet{paivaComparisonBinlessSpike2010} and \citet{sihnSpikeTrainDistance2019} for a comparison of different spike train measures. 

Each of the three spike train metrics used in this work is parameterized by a smoothing parameter. For Van Rossum distance and Schreiber similarity, their parameters are an integral component of their definition. We augment the Pearson correlation by prepending a smoothing step copied from the Schreiber similarity. For Schreiber similarity and Pearson correlation, spike trains are first smoothed with a Gaussian kernel with standard deviation parameter $\sigma$; for Van Rossum distance, the first step is to smooth spike trains with an exponential kernel with decay parameter $\tau$.

An important distinction between Van Rossum distance and Schreiber similarity is their sensitivity to the overall number of spikes. Van Rossum distance lacks any normalization to spike count—it cares about both where spikes are placed and how many. In comparison, Schreiber similarity includes a normalization to the number of spikes present in both spike trains. This distinction between the metrics explains how some models can perform well in terms of one metric and poorly in terms of the other. To illustrate this, consider an example where a cell's average spike rate is known, but there is no information on when spikes occur. Two prediction strategies are: predict zero spikes or predict a constant rate of spikes equal to the known average rate. Van Rossum distance will score the first strategy more favourably than the second, while it is the reverse for Schreiber similarity. The Van Rossum distance will penalize the many incorrectly timed spikes predicted by the second strategy. Under Schreiber similarity, the first strategy will necessarily score zero while the second strategy scores positively for having some positive projection onto the ground truth spike train. This distinction explains how, if the preponderance of a Poisson model's outputs are close to zero, switching to sampling for spike count inference can lead to better performance in terms of Schreiber similarity and worse performance in terms of Van Rossum distance. This illustrates the benefit of using both Van Rossum distance and Schreiber similarity for evaluation: they give different perspectives on the quality of a predicted spike train.

\subsection{Evaluation Methodology} \label{app:evaluation_methodology}
In this section, we justify our departure from other evaluation practices.

This work's approach to evaluation is to consider the quality of generated spike trains. Our evaluation departs from previous work by avoiding binning, avoiding a single fixed smoothing, and avoiding multi-trial averaging. In \citet{battyMultilayerRecurrentNetwork2016}, \citet{mcintoshDeepLearningModels2016} and \citet{cadenaDeepConvolutionalModels2019}, ground truth spikes are \textit{binned} then Gaussian \textit{smoothed} and compared to Gaussian \textit{smoothed} model outputs. Furthermore, ground truths are \textit{averaged} over multiple trials in \citet{mcintoshDeepLearningModels2016} and \textit{normalized} by the average over multiple trials in \citet{battyMultilayerRecurrentNetwork2016} and \citet{cadenaDeepConvolutionalModels2019}. In each of these studies, the choice of bin size and smoothing filter width were \textit{fixed} and chosen to match the bin size and smoothing used during training. A downside of our approach is that, by choosing spike trains as the inputs to comparison, we must first convert Poisson model outputs into spike predictions, and this introduces the extra consideration of which spike inference strategy works best.

We also eschew averaging or normalizing across repetitions of the same stimulus. As argued by \citet{brettePhilosophySpikeRateBased2015}, neuron state is not stable, and differences across stimulus presentations may arise not solely from noise to be disregarded but from the fact that the state of a neuron or network of neurons can change. With this in mind, we present a single stimulus pattern without repetition and delegate the modelling of uncertainty and identification of latent variables to the neural network.

Finally, we argue against the use of negative log-likelihood (NLL) as an effective evaluation metric. Likelihood-based metrics are common in neural temporal point process literature. In their review of neural temporal point processes, \citet{shchurNeuralTemporalPoint2021a} argue against the prevalent usage of NLL, saying that ``NLL is mostly irrelevant as a measure of error in real-world applications'' and claim that the lack of investigations focusing on domain-specific metrics is a major gap in current neural temporal point process literature. In our work, NLL as a metric would not allow the direct comparison of Poisson models with different summation intervals, nor would it provide a way to measure the performance of the spike distance model, which does not output probability distributions. Our results demonstrate the benefits that can be achieved by working outside the probabilistic setting.

\subsection{Retina Recording} \label{app:recording}
The main dataset is created from a single 15-minute recording of chicken \ac{RGC} spike activity carried out by \citet{seifertBirdsMultiplexSpectral2023}. Important aspects of the recording will be summarized here (see \citet{seifertBirdsMultiplexSpectral2023} for full details). A chicken retina was placed on an \ac{MEA}, and a full-field colour sequence driven by 4-LEDs was projected onto the retina. The colour sequence was a random sequence of \qty{50}{ms} frames, with each LED having a 50\% chance of being on or off in each frame. The electrical activity recorded by the \ac{MEA} was passed through a spike sorter to produce a sequence of spike events sampled at \qty{17.9}{kHz}. The stimulus frame changes were also recorded as trigger events at \qty{17.9}{kHz}. At this point, quality criteria were used to filter out cells which were considered to respond poorly. The stimulus was downsampled by a factor of 18 to \qty{992}{Hz}, and the spike data was rebinned to \qty{992}{Hz}. Keeping correspondence between the sample period of the stimulus and spike data allows both to be stacked into a single 2D array. Downsampling by 18 was chosen as the sample period of \textasciitilde \qty{1}{ms} is short enough that no two spikes are recorded in the same bin, which simplifies the spike distance calculation. Downsampling by 18 is also convenient as it results in a sample period very close to \qty{1}{ms} (\qty{1.0008}{\ms}), enabling number of samples and number of milliseconds to be used interchangeably with little effect on precision. 

The same procedure was followed to collect the frog \ac{MEA} recording, except for one difference: no cell filtering based on quality criteria was carried out.

The spike trains produced by the spike sorter are not a perfect representation of the true spike activity produced by the neurons—false positives, false negatives and the incorrect assignment of spikes to cells are all issues.  In this work, we consider the neurons, the multi-electrode array and the spike sorter as a single black box whose behaviour we predict; we do not attempt to make claims about any of the three components individually. 

Further processing of the data was carried out in order to create separate training, validation and test sets. This is described next.

\subsection{Dataset} \label{app:dataset}
The 15-minute recording was split according to the ratio (7, 2, 1) into training, validation and test sets as shown previously in Figure \ref{fig:data}. As the health of the retina can degrade over the course of the recording altering its behaviour, the (7, 2, 1) split was formed by first splitting the recording by the ratio (7, 2, 2, 2, 7) and then combining segments: the first and last 5.25-minute segments form the training set, the single middle 90-second segment forms the test set, and the two remaining 90-second segments form the validation set. This approach allows the test set to be exposed to both extremes while keeping the test set as a single contiguous chunk.

For the Poisson models, each dataset sample is a tuple (input, target), where the input is a $5 \times 992$ array representing 1 second of stimulus and spike history, and the target is a scalar representing the number of spikes in the summation interval. For the spike distance model, each sample is a tuple (input, target) where the input is the same as that for the Poisson models, and the target is a $128$-length array representing \textasciitilde \qty{128}{\ms} of ground-truth spike distance data, with $t=0$ positioned at index 32. The choice of these values are hyperparameters discussed in Appendix \ref{app:hyperparameters}.

As the dataset's sequential nature means that adjacent samples will be the same except for a 1 time step shift, we introduce a configurable stride parameter for the training dataset. A larger stride decreases the number of samples that constitute an epoch; this allows model evaluation and checkpoints to be carried out frequently while still being synchronized to epoch completions. The intermediate samples are still used in training, but they appear as augmentations by random indexing between stride indices. For all experiments, the stride was set to 13. A value of 13 resulted in the U-shaped learning curve being observed over 80 epochs. 

\newcolumntype{L}{>{\centering\arraybackslash}m{0.20\linewidth} }
\begin{table*}[t]
\centering
\caption{Base architecture shared by both Poisson and spike distance models. \textasciitilde \qty{302}{\kilo\nothing} parameters, \textasciitilde \qty{8.28}{\mega\nothing} multi-adds.}
\addtolength{\tabcolsep}{-2pt}
\vspace{4ex}
\scalebox{0.85}{
  \begin{tabular}{L|c|c|c}
Layer & Input size & Kernels (length, channels, stride) & Output size \\
\hline
\hline
\multirow{3}{*}{Initial conv} & 
\multirow{3}{*}{5$\times$992} & 
\multirow{3}{*}{$\begin{bmatrix}15, & 64, & 2\end{bmatrix} \times 1$} & 
\multirow{3}{*}{$64 \times 496$} \\
& & & \\
& & & \\
\hline 
\multirow{6}{0.2\textwidth}{ConvNext blocks \\ (downsampling)} & 
\multirow{6}{*}{\begin{tabular}[c]{@{}c@{}} 64 $\times$ 496 \end{tabular}} & 
\multirow{6}{*}{$\begin{bmatrix}3 & 64 & 2\\1 & 128 & 1\\5 & 128 & 1\\1 & 64 & 1\end{bmatrix}$ $\times$ 6}  & 
\multirow{6}{*}{$64 \times 8$} \\
& & & \\
& & & \\
& & & \\
& & & \\
& & & \\
\hline
\multirow{5}{*}{ConvNext blocks} &
\multirow{5}{*}{\begin{tabular}[c]{@{}c@{}} $64 \times 8$ \end{tabular}} & 
\multirow{5}{*}{$\begin{bmatrix}1, & 128, & 1\\3, & 128, & 1\\1, & 64, & 1\end{bmatrix}\times 4$}  & 
\multirow{5}{*}{$64 \times 8$} \\
& & & \\
& & & \\
& & & \\
& & & \\
\hline

\end{tabular}
}

\label{table:base_model}
\end{table*}

\begin{table*}[t]
\centering
\caption{Distance model head. \textasciitilde\qty{29}{\kilo\nothing} parameters, \textasciitilde\qty{587}{\kilo\nothing} multi-adds.}
\addtolength{\tabcolsep}{-2pt}
\vspace{4ex}
\scalebox{0.85}{
  \begin{tabular}{L|c|c|c}
Layer & Input size & Kernels (length, channels, stride) & Output size \\
\hline
\hline
\multirow{5}{*}{ConvNext} & 
\multirow{5}{*}{\begin{tabular}[c]{@{}c@{}} 64 $\times$ 8 \end{tabular}} & 
\multirow{5}{*}{$\begin{bmatrix}1, & 128, & 1\\5, & 128, & 1\\1, & 16, &
1\end{bmatrix} \times 1$}  & 
\multirow{5}{*}{16 $\times$ 16} \\
& & & \\
& & & \\
& & & \\
& & & \\
\hline

\multirow{5}{*}{ConvNext} & 
\multirow{5}{*}{\begin{tabular}[c]{@{}c@{}} $16 \times 16$ \end{tabular}} & 
\multirow{5}{*}{$\begin{bmatrix}1, & 32, & 1\\5, & 32, & 1\\1, & 16, & 1\end{bmatrix} \times 3$}  & 
\multirow{5}{*}{$16 \times 128$} \\
& & & \\
& & & \\
& & & \\
& & & \\
\hline

\multirow{3}{*}{Conv} & 
\multirow{3}{*}{$16\times128$} & 
\multirow{3}{*}{$\begin{bmatrix}1, & 1, & 1\end{bmatrix} \times 1$} & 
\multirow{3}{*}{$1 \times 128$} \\
& & & \\
& & & \\
\hline
\end{tabular}
}
\label{table:distance_model}
\end{table*}
\begin{table*}[!t]
\centering
\caption{Poisson model head. 513 parameters, 513 multi-adds.}
\addtolength{\tabcolsep}{-2pt}
\vspace{4ex}
\scalebox{0.85}{
\begin{tabular}{c|c|c|c}
Layer & 
Input size & Kernels (length, channels, stride) & Output size \\
\hline
\hline

\multirow{3}{*}{Flatten} & 
\multirow{3}{*}{\begin{tabular}[c]{@{}c@{}}$ 64 \times 8 $ \end{tabular}} & 
\multirow{3}{*}{None} &
\multirow{3}{*}{512} \\
& & & \\
& & & \\
\hline

\multirow{4}{*}{FC layer} & 
\multirow{4}{*}{\begin{tabular}[c]{@{}c@{}} $512 \times 1 $\end{tabular}} & 
\multirow{4}{*}{$\begin{bmatrix}1, & 512, & 1\end{bmatrix}$}  & 
\multirow{4}{*}{1} \\
& & & \\
& & & \\
\hline
\end{tabular}
}
 
\label{table:poisson_model}
\end{table*}

\subsection{Models} \label{app:models}
The convolutional kernels used in the shared base architecture are shown in Table \ref{table:base_model}. In addition to the convolutional layers, a learnable positional embedding was added to the output of the initial convolutional layer of the shared architecture. The kernels used for the spike distance head and the Poisson head are shown in Tables \ref{table:distance_model} and \ref{table:poisson_model} respectively. The shared base architecture accounts for 88\% of parameters and 93\% of multi-adds in the spike distance model, and accounts for practically all the parameters and multi-adds for the Poisson models. 

ConvNext blocks introduced by \citet{liuConvNet2020s2022} were chosen to form the main component of the architecture as they are commonly used and well-tested. Our implementation includes the Global Response Normalization layer (GRN) from ConvNextv2 blocks introduced by \citet{wooConvNeXtV2CoDesigning2023}. We use dropout at the end of each block, at a rate of 0.2.

\subsection{Architecture Search} \label{app:arch_search} 
The number of mid-layers in the shared base architecture was varied in order to increase confidence that performance differences between the spike distance model and the Poisson models are not due to differences in layer and parameter counts. The number of mid-layers was varied between 0 and 7 (inclusive) and the configuration that best served the Poisson-\qty{80}{\ms} model was chosen. This selection was based on the interquartile mean of the loss across cells in the validation set. Focus was given to the Poisson-\qty{80}{\ms} model as it was observed to be the most competitive Poisson model.

For each of these configurations, the Poisson-\qty{80}{\ms} model was trained once for each of the 60 chicken \acp{RGC}. The interquartile mean loss is shown in Figure \ref{fig:loss_by_nlayer}. To gauge the magnitude of the effect on metrics, the interquartile mean Pearson correlation is reported in Figure \ref{fig:pcorr_by_nlayer}. Other metrics show a similar lack of dependency on layer count. The 5-layer configuration was selected based on the loss; however, it seems like the number of mid-layers has very little effect on the performance of the Poisson-\qty{80}{\ms} model.

\clearpage
\begin{figure*}[htp]
    \begin{minipage}[t]{0.495\textwidth}
        \centering
      \vspace{0pt}
      \includegraphics[width=\textwidth]{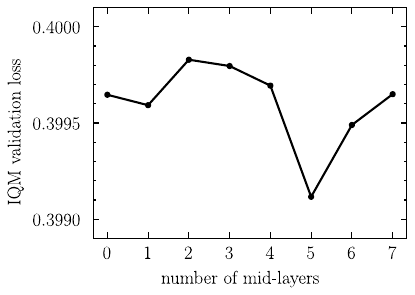}
        \caption{Interquartile mean validation loss for a range of mid-layer counts. The IQM is calculated over the 60 separately trained Poisson-80 models, one for each of the 60 \acp{RGC}. The loss is the mean validation set loss.}
        \label{fig:loss_by_nlayer}
    \end{minipage}
    \hfill
    \begin{minipage}[t]{0.470\textwidth}
      \centering
      \vspace{0pt}
    \includegraphics[width=\textwidth]{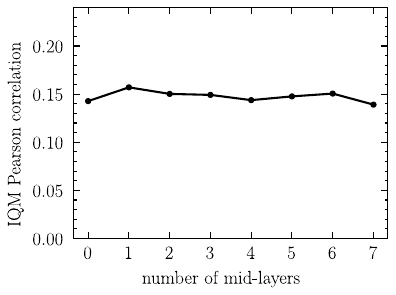}
        \caption{Interquartile mean Pearson correlation for a range of mid-layer counts. The IQM is calculated over the 60 separately trained Poisson-80 models, one for each of the 60 \acp{RGC}. Pearson correlation is calculated between the output and validation set ground-truth spike trains, smoothed with a Gaussian kernel with $\sigma=20$.}
        \label{fig:pcorr_by_nlayer}
    \end{minipage}
\end{figure*}

\begin{figure*}[ht]
  \centering
  \includegraphics[width=\textwidth]{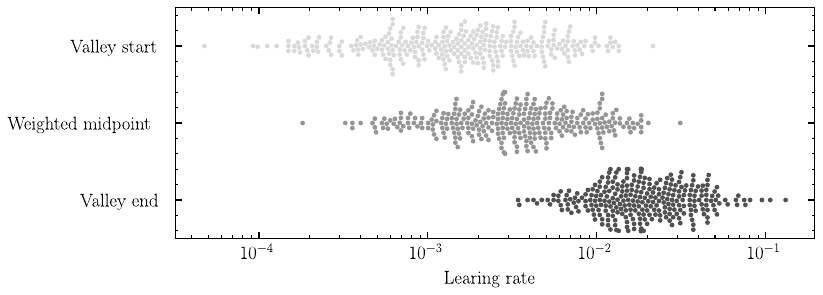}
  \caption{Three learning rates identified by the range test for learning rates described by \citet{smithSuperconvergenceVeryFast2019}. Valley start and valley end describe the edges of the U-shaped valley containing the suggested learning rate, which is a weighting of the edges. }
    \label{fig:lr_finder}
\end{figure*}

\clearpage

\subsection{Hyperparameters} \label{app:hyperparameters}
A list of hyperparameters associated with models, training and inference.

\textbf{Architecture}. The following settings are associated with the model architectures.
\begin{itemize}
  \item Channel count in the shared base: 64. Early experiments not reported here found 32 to have slightly worse performance for all models, and 64 and 128 were observed to have similar performance. 
  \item Channel count in the distance model head: 16. Increasing this setting was observed to improve performance; however, to keep the model's parameter count close to that of the Poisson models, the channel count was kept low.
  \item Channel expansion factor of the ConvNext blocks: 2. Changes to this value were not tested.
  \item Dropout rate: 0.2. Early experiments not reported here observed benefits of dropout for all models. Once added, the dropout rate was not experimented with further.
\end{itemize}

\textbf{Training}. 
Settings associated with training were chosen as described below. 

\begin{itemize}
  \item Maximum learning rate: $\num{5e-4}$. This was chosen by running the LR range test introduced by \citet{smithSuperconvergenceVeryFast2019} on an earlier architecture for all chicken \acp{RGC}. After finalizing the architecture (see Appendix \ref{app:arch_search}), the range test was run again to ensure that the learning rate was still appropriate. Figure \ref{fig:lr_finder} plots the results of this second range test, which supports $5 \times 10^{-4}$ still being an appropriate learning rate.
  \item Batch size: 256. Smaller batch sizes resulted in a training speed considered onerously slow.
  \item Dataset stride: 13. The dataset stride is described in Appendix \ref{app:dataset}. 
  \item Epochs: 80. Training is carried out for 80 epochs. Each time step will be part of a model input on average $80 \times \frac{992}{13} = 6105$ times.
  \item AdamW parameters: $(\beta_1, \beta_2, \text{eps}, \text{weight decay}) = (0.9, 0.99, \num{1e-5}, 0.3)$. These were chosen in line with heuristics described by \citet{howardDeepLearningCoders2020}.
  \item Learning rate scheduler: 1-cycle scheduler policy described by \citet{smithSuperconvergenceVeryFast2019} was used, with three-phase enabled and other options as default values in Pytorch 2.0's implementation.
  \item Precision: Nvidia's automatic mixed precision was used.
\end{itemize}

\textbf{Spike distance}. A number of settings were used to control the nature of the output spike distance array and the inference process:

\begin{itemize}
  \item \textit{stride}: how many time steps to shift forward when carrying out autoregressive inference. This was set to 80 time steps (\textasciitilde \qty{80}{ms}). Shorter strides offer improved inference at the cost of increased compute. 80 was chosen to match the Poisson-\qty{80}{\ms} model.
  \item $L$: the length of the spike distance array the model outputs. 128 is the next power of two greater than 80, allowing the distance model's head to be a simple sequence of 4 upsample blocks. The additional time steps were split between before $t=0$ and after $t=\qty{80}{ms}$ according to the next setting, $i_0$.
  \item $i_o$: the index into the array that corresponds to $t=0$. This was set to 32. A low value negatively affects inference as described in Appendix \ref{app:extended_dist_arr}. 32 was chosen as it is the length of 2 activations before passing through the 4 upsample blocks (leaving 1 activation that extends past $t=80$). A search for an optimal value has not been carried out.
  \item $M$: the maximum distance for an element of the spike distance field. This was set to \qty{200}{ms}. This is not a very sensitive parameter. A few settings were experimented with, as described in Appendix \ref{app:maximum_distance}.
\end{itemize}

\subsection{Computational Resources} \label{app:comp_resources}

\begin{figure*}[ht]
  \centering
  \includegraphics[width=\textwidth]{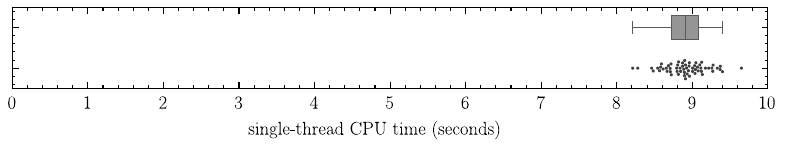}
  \caption{Running time of Algorithm \ref{alg:spike_inference} to infer 90 seconds of spike behaviour for 60 cells, in steps of \qty{80}{\ms}. Running time is calculated with Python's \texttt{timeit} package. In order to isolate the running time of Algorithm \ref{alg:spike_inference}, inference is not autoregressive, and instead uses precalculated model outputs as inputs to the algorithm. (Mean, standard deviation) = (8.9, 0.27).} \label{fig:inference_perf}
\end{figure*}

All training and inference was carried out on a single workstation, with CPU, GPU and RAM specifications: AMD Ryzen 9 5900X CPU, Nvidia RTX 3090 GPU and 128 GiB RAM.

Training a single model on all 60 chicken \acp{RGC} took \textasciitilde \qty{12}{hours}. This was repeated 11 times for the Poisson-\qty{80}{ms} model and 11 times for the spike distance model in order to calculate confidence intervals. Only 1 repeat was carried out for the remaining five Poisson models. For architecture variations, the Poisson-\qty{80}{\ms} model was trained 8 times for all \acp{RGC} to investigate the number of mid-layers. Training on the frog \acp{RGC} took \textasciitilde \qty{23}{hours}, and was done once for each of the 7 models. 

The total training time was \textasciitilde \qty{574}{hours}:
  \qty{12}{hours\per model} $\times$ \qty[parse-numbers=false]{(2 $\times$ 11 + 5 + 8)}{models}
  + \qty{23}{hours\per model} $\times$ \qty{7}{models} = 581.
Inference was considerably quicker, taking a total of \textasciitilde 8 hours summed over all experiments.

\subsection{Spike Distance} \label{app:spike_distance}
This section covers some details on working with spike distance: choosing a maximum spike distance, and the benefits to inference of outputting an extended spike distance.

\subsubsection{Maximum Spike Distance} \label{app:maximum_distance}
If there are no spikes, then the distance to the nearest spike can be considered infinite. When approximating a distance function, infinite distances can be avoided by fixing a maximum allowed distance. The choice of a maximum is important beyond simply the avoidance of undefined values—it is also useful to limit how far into the past or future we expect a distance function to be reasonably approximated. Consider the spike distance function $f_S$ evaluated at some time point $t_1$. The single value $f_S(t_1)$ contains information about all time points in the interval $[t_1 - f_S(t_1), t_1 + f_S(t_1)]$. For example, if $f_S(t_1) = 10$, then we know that there are no spikes in the open interval $(t_1 - 10, t_1 + 10)$. When training a model for spike prediction, we do not expect the model to anticipate spiking activity far into the future. For a model that predicts spike distance, one way to limit its exposure to the future is to clamp the target spike distance function. In this paper, we work with a maximum spike distance of \qty{200}{ms}. Through experimentation not reported here, for the dataset used in this paper, we tested a range of maximum distances and found that there was not a major difference in performance when the maximum spike distance is between \qty{100}{ms} and \qty{600}{ms}. Below \qty{100}{ms} and above \qty{600}{ms} performance begins to degrade.

\subsubsection{An Extended Spike Distance Improves Inference} \label{app:extended_dist_arr}
Spike prediction performance is improved if the spike distance function outputted by the neural network extends both before and after the interval in which spikes are to be predicted. This means that when predicting $L$ time steps starting from $t_0$, the neural network should output a spike distance function that extends both before $t_0$ and after $t_0 + L$. 

The reason for this lies in the capacity of spike distance to carry information about a spike train both before and after the time point at which it is evaluated. For example, the spike distance at $t_0 - 1$ (past) can be affected by spikes that come after $t_0$ (future). In other words, the spike distance representation for a spike train does not fall neatly into the same time interval. By extending the spike distance function before $t_0$ and after $t_0 + L$, there is more information available for the inference algorithm to judge the placement of spikes. The model presented in this work is trained to output a spike distance array of length $128$, where $32$ elements are assigned to the past. If predicting a spike train for 80 time steps into the future, then there would be $16$ elements of the model's output remaining that extend beyond the end of the spike train.

\subsection{Discrete Spike Distance} \label{app:discrete_distance}
In practice, both spike times and spike distance will be discretized. Instead of the input being exact spike times, we will use spike counts sampled at a certain sample period $T$, and instead of evaluating the spike distance at any real, we will evaluate the distance at discrete times separated by the same period $T$ and collect the values in a spike distance array.

A naive and computationally simple approach of approximating the continuous spike distance is to count the number of samples until a sample with a non-zero spike count. This is a suitable approach when the sample rate is fast in comparison to the spike rate and spikes are rarely coincident in the same sample. Indeed, the main results in this work involve sample frequencies fast enough that no spikes are coincident in the same sample. When using a slower sampling rate, it may be beneficial to use the approach outlined below, which more faithfully discretizes the minimum distance when there is uncertainty over the precise location of spikes. 

To approximate the continuous spike distance function $f_S(t) : \mathbb{R} \to \mathbb{R}$ introduced in Section \ref{sec:distance}, we are seeking a function $\hat{f}_S(t) : \mathbb{N} \to \mathbb{R}$ that is a faithful approximation when both spike times and the points being evaluated are discretized. \citet{osherLevelSetMethods2003} presents several methods for discretizing implicit representations; there the authors use known properties of the physical systems being modelled in order to discretize implicit functions. Our spike events are far less analytically constrained, and so we will instead take a probabilistic approach.

\subsubsection{Motivating the Discrete Spike Distance}
In the continuous case, the two points between which a distance is calculated are obvious: the spike time and the time for which the spike distance is being evaluated. In the discrete case, we must consider candidates for both of these time points.

Consider recording spikes at a sample period $T = \qty{20}{ms}$, where each sample is a sum of the spikes that occurred within a $\qty{20}{ms}$ interval. If a single spike is recorded in the 5\textsuperscript{th} sample, and the 5\textsuperscript{th} sample collected spikes from the interval $(\qty{80}{ms}, \qty{100}{ms})$, then where in this interval should we consider the spike to have occurred? Another question: if we wish to calculate an array of spike distances at some sample rate, what reference point should each sample use from which to measure distance? For spike times there is uncertainty, and for reference points there is a choice to be made. 

A few decisions and assumptions are sufficient to determine a discrete distance. First, we will treat spike times as being random variables distributed uniformly over the sample in which they are recorded. This decision is reasonable when the samples represent a count of spikes within an interval, and there is no further information about the location of spikes within the interval. Indeed, this is the case for this work—the spike events are outputs of a  black-box detection mechanism of the \ac{MEA} and are subsequently rebinned to reduce the sample rate. The second decision is to fix our array of distance calculations to the same frequency as the spike recording, and we measure distances from the midpoint of each interval.

\begin{figure}
  \centering
  \includegraphics[width=0.7\columnwidth]{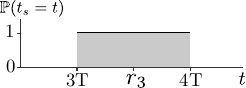}
  \caption{If one spike is detected in sample 3, then the spike distance for sample 3 is given by $\hat{f}(3) = \mathbb{E}[t_3 - r_3] = 0.25T$, where $r_3 = 3.5T$ and $t_3$ is a random variable uniformly distributed over $[3T, 4T]$.}
    \label{fig:spikeprob}
\end{figure}

To appreciate the implications of these choices, consider the case where a single spike is recorded in sample $i$ and we wish to evaluate the discrete spike distance for the same index, $i$. The naive solution is to define $\hat{f}_S(i) = 0$. However, 0 is not a good representation of the spike distance over the whole interval $[iT, (i+1)T]$; 0 is only ever the value of the spike distance at the instantaneous time a spike occurs. A better value for $\hat{f}_S(i)$ should incorporate the fact that there is uncertainty about when exactly a spike occurs within this interval of length $T$. The approach we take goes as follows. Fix the reference point $r_s$ from which spike distances are measured to be the interval's midpoint, $r_s = (i+0.5)T$. Consider the spike time to have a uniform distribution over the interval $[iT, (i+1)T]$. The value for $\hat{f}_S(i)$ will be defined as the expected distance to the reference point, where the expectation is calculated over the uniformly distributed spike time. Figure \ref{fig:spikeprob} shows this case in more detail. The next section formalizes these ideas and arrives at a simple procedure for calculating the discrete spike distance from a sequence of spike counts.

\subsubsection{Definition and Calculation}
What follows is a formalization of the discrete spike distance that can be thought of as applying extreme value theory to spike distances. \citet{colesIntroductionStatisticalModeling2001} is a good introduction to extreme value statistics. What prevents the direct use of a plug-and-play theorem from extreme value theory is that in our situation we are not interested in the asymptotic behaviour of an unknown distribution, but the finite behaviour of a known distribution.

As part of the formalization, we will also generalize what was introduced in Section \ref{sec:distance}. Equation \ref{eq:energy} in Section \ref{sec:distance} defined a spike count sequence to be a binary array. This was a simplification allowed for by the high sample rate that ensured no two spikes were recorded in the same sample. Below, we relax this assumption and allow for the possibility of multiple spikes in a single sample.

We set the scene with two definitions. Similar to how the continuous spike distance function was parameterized by a set of spike times, the discrete spike distance function is parameterized by a sequence of spike counts. 

\begin{definition}Let $T \in \mathbb{R}$ be the sample rate. 
  Let $l \in \mathbb{N}$ be a positive integer representing the number of
  samples. Then a \textbf{spike count sequence} $S = (s_i)_{i=0}^{l-1}$ at 
  sample rate $T$ is a sequence of spike counts, where $s_i$ is the number of
  spikes recorded in the interval $[iT, (i+1)T]$. 
\end{definition}

If $S = (s_i)_{i=0}^{l-1}$ is a spike count sequence with a total spike count $n = \sum_{i=0}^{l-1}s_i$, we can associate with it a sequence of independent random variables, $X_0, X_1, \ldots X_{n-1} $, where each $X_j$ represents a spike time distributed uniformly over one of the length $T$ intervals. For example, if $s_0 = 3$ and $s_1 = 1$, then $X_0, X_1$ and $X_2$ would be distributed uniformly over the interval $[0, T]$ while $X_3$ would be uniformly distributed over the interval $[T, 2T]$.

\begin{definition} Let $i$ be a sample index. 
  Let $S$ be a spike count sequence and let $X_0, X_1, ...X_{n-1}$ be independently identically distributed real random variables representing the spike times of the $n$ spikes from $S$. Each $X_j$ is distributed uniformly over the sample interval in which it is recorded. Let $r_i = i + 0.5$ represent the midpoint of the $i^\text{th}$ sample interval. For each $X_j$ let $D_j = |X_j - r_i|$ be the derived random variable representing the distance of the $j^{\text{th}}$ spike to the midpoint $r_i$. The \textbf{discrete spike distance function} $\hat{f}_S: \mathbb{N} \to \mathbb{R}$ at $i$ is given by:

\begin{equation}
  \hat{f}_{S}(i) = \mathbb{E}_{D_0, D_1, \ldots, D_{n-1}} \min_{0 \leq j < n} \; D_j \;
  \label{eq:discrete_dist}
\end{equation} 

and has units of the sampling period, $T$.
\end{definition}

Calculating this value involves considering only the one or two samples containing spikes that are closest to the sample of interest, $i$. This could be sample $i$ itself, a single sample to the left or right of $i$ or two samples equidistant to $i$, one on either side. Being able to ignore most of the samples allows for the discrete spike distance to be expressed neatly in terms of the spike counts of the closest non-empty sample(s).

\begin{proposition}
Let $d \in \mathbb{N}$ be the difference in samples between the sample of interest $i$ and the closest spike-containing sample(s). Let $m \in \mathbb{N}$ the number of spikes contained within the closest spike-containing sample(s).

The value of $\hat{f}_S(i)$ from Equation \ref{eq:discrete_dist} is given by: 

\begin{equation}
\hat{f}_S(i) =
\begin{cases}
  \hat{f}_S(i) = \frac{1}{2(m + 1)} & \parbox{3cm}{if one or more spikes in sample $i$}
  \vspace{6pt} \\
  \hat{f}_S(i) = d -\frac{1}{2} + \frac{1}{m+1} & \text{otherwise.}
\end{cases}
  \label{eq:discrete_dist_calc}
\end{equation}
\end{proposition}

\begin{proof} First, the first case: there are $m>0$ spikes in sample $i$.
The presence of the $\min$ function in Equation \ref{eq:discrete_dist} means that we can ignore all spikes that are not in sample $i$, as they will not affect the expectation. Let $D_0, D_1, \ldots, D_{m-1}$ be the random variables representing the distance of the $m^{th}$ spike recorded within sample $i$. Single out one of these, $D_1$, and we will return later to include the others. $D_1$ will contribute to the expectation only when it is smaller than all others: if $D_1 = t$, then all $m-1$ other distances must be greater than $t$. The probability $\mathbb{P}(D_1 = t) = 2$ (uniform over interval of length $\frac{1}{2}$) and the probability of other $D_j$ having a value larger than $t$ is $\mathbb{P}(D_j > t) = 2(\frac{1}{2} - t)$.

The contribution to the expectation for $D_1$ then amounts to:

\begin{align*}
& \int_{0}^{\frac{1}{2}} t \cdot \mathbb{P}(D_1=t) \prod_{j=1}^{m-1} \mathbb{P}(D_j > t) \dd{t} \\
&=  \int_{0}^{\frac{1}{2}} t \cdot (2) \prod_{j=1}^{m-1} (2) (\frac{1}{2} - t)\dd{t} \\
&=  2 \int_{0}^{\frac{1}{2}} t (1 - 2t)^{m-1} \dd{t} \\
&= 2 \left[-\frac{t}{m}(1-2t)^m\right]_{0}^{\frac{1}{2}} +
\int_{0}^{\frac{1}{2}} \frac{1}{m}(1-2t)^m \dd{t} \\
&= 0 + \left[ - \frac{1}{2m(m+1)} (1-2t)^{m+1}
\right]_{0}^{\frac{1}{2}} \\
&= \frac{1}{2m(m+1)} \;.
\end{align*}
We singled out $D_1$; however, any of the $m$ spikes are equally likely to
be the closest spike. Summing the contribution over $m$ spikes gives us:
\[
\frac{1}{2(m+1)}
\]
as required by the first case.

What distinguished the second case is that the reference point $r_i$ is not in the same sample as the closest spike(s). Let $m > 0$ be the number of spikes that are in the closest one or two samples, situated $d$ samples from the sample of interest $i$. We ignore spikes from other samples. Let $D_0, D_1, ... D_{m-1}$ be the random variables representing the distance of the $m^{th}$ spike to the reference point $r_i$. Single out one of these, $D_1$. $D_1$ will contribute to the expectation only when it is smaller than all others: if $D_1 = t$, then all $m-1$ other distances must be greater than $t$. This time, $D_1$ has a uniform distribution over the interval $[d - \frac{1}{2}, d + \frac{1}{2}]$ of length $1$ with $\mathbb{P}(D_1 = t) = 1$ in this interval. The probability of other $D_j$ having a value larger than $t$ is $\mathbb{P}(D_j > t) = d + \frac{1}{2} - t$.

\begin{table*}[t]
    \centering
    \caption{The approximate spike distance compared to the discrete spike distance. Discrete spike distances are evaluated for the samples 0 to 8 when there are three spikes recorded: one in sample $2$ and two in sample $8$.}
    \vspace{0.3cm} 
    \renewcommand{\arraystretch}{1.7}
    \newcolumntype{M}{>{\centering\arraybackslash} p{0.40cm}}
    \begin{tabular}{|m{5cm}|M|M|M|M|M|M|M|M|M|}
        \hline
        sample index, $i$ & $0$ & $1$ & $2$ & $3$ & $4$ & $5$ & $6$ & $7$ & $8$ \\
        \hline
        spikes & & &  \cellcolor{gray!25} $|$ & & & & & & \cellcolor{gray!25} $| |$ \\
        \hline
        $m$ & $1$ & $1$ & $1$ & $1$ & $1$ & $3$ & $2$ & $2$ & $2$ \\
        \hline
        $d$ & 2 & 1 & $0$ & 1 & 2 & 3 & 2 & 1 &  $0$ \\
        \hline
        spike distance, $\hat{f}_S(i)$, in units of $T$ & $2$ & $1$ & 
        $\frac{1}{4}$ & $1$ & $2$ & $2 \frac{3}{4} $ & $1 \frac{5}{6} $ &
        $\frac{5}{6}$ & $\frac{1}{6}$ \\
        \hline
        approx. spike distance, in units of $T$ & $2$ & $1$ & 
        $0$ & $1$ & $2$ & $3 $ & $2 $ & $1$ & $0$ \\
        \hline
    \end{tabular}
    \label{tab:discrete_comparison}
\end{table*}

The contribution to the expectation for $D_1$ then amounts to:

\begin{align*}
& \int_{d-\frac{1}{2}}^{d+\frac{1}{2}} t \cdot \mathbb{P}(D_1=t) \prod_{j=1}^{m-1} \mathbb{P}(D_j > t) \dd{t} \\
&=  \int_{d-\frac{1}{2}}^{d+\frac{1}{2}} t \cdot (1) \prod_{j=1}^{m-1}
(d + \frac{1}{2} - t) \dd{t} \\ 
& \text{and by change of variable,} \\
&=  \int_{0}^{1} (d - \frac{1}{2} + t) \cdot (1) \prod_{j=1}^{m-1}
(1 - t)\dd{t} 
\\
              &=  \int_{0}^{1} (d - \frac{1}{2} + t )(1 - t)^{m-1} \dd{t} \\
              &= \left[-\frac{d - \frac{1}{2} + t}{m}(1-t)^m\right]_{0}^{1} +
                           \int_{0}^{1} \frac{1}{m}(1-t)^m \dd{t} \\
              &= \frac{d - \frac{1}{2}}{m} + \left[ - \frac{1}{m(m+1)} (1-t)^{m+1}
                           \right]_{0}^{1} \\
              &= \frac{d - \frac{1}{2}}{m} + \frac{1}{m(m+1)} \;.
\end{align*}

We singled out $D_1$; however, any of the $m$ spikes are equally likely to
be the closest spike. Summing the contribution over $m$ spikes gives us:
\[
d - \frac{1}{2} + \frac{1}{m(m+1)}
\]
as required by the second case.
\end{proof}

Table \ref{tab:discrete_comparison} compares the discrete spike distance described above to the approximate spike distance, including the values for $m$ and $d$ from Equation \ref{eq:discrete_dist_calc}.

\end{document}